%% file: SIADS_monoid_formalism.tex
\documentclass[final]{siamonline190516}
\usepackage[utf8x]{inputenc}

\PrerenderUnicode{ç}
\PrerenderUnicode{’}

\usepackage{amsfonts}
\usepackage{graphicx}
\usepackage{tikz}
\usepackage{subcaption}

\usepackage{algorithm}
\usepackage{algorithmic}

\usetikzlibrary{calc}

\newcommand\DoubleLine[7][4pt]{%
	\path(#2)--(#3)coordinate[at start](h1)coordinate[at end](h2);
	\draw[#4]($(h1)!#1!90:(h2)$)-- node [auto=left] {#5} ($(h2)!#1!-90:(h1)$); 
	\draw[#6]($(h1)!#1!-90:(h2)$)-- node [auto=right] {#7} ($(h2)!#1!90:(h1)$);
}

\newcommand{\bp}{{\bowtie}}

\newcommand{\xset}{\mathbb{X}}

\newcommand{\yset}{\mathbb{Y}}
\newcommand{\wset}{\mathbb{W}}
\newcommand{\mset}{\mathbb{M}}

\newcommand{\type}[1]{\mathcal{T}_\mathcal{#1}}
\newcommand{\tset}{T}

\newsiamremark{exmp}{Example}
\newsiamremark{remark}{Remark}

\title{Commutative monoid formalism for weighted coupled cell networks and invariant synchrony patterns\thanks{Submitted to the editors at 21 of December of 2020.
		\funding{This work was supported in part by the FCT Project IMPROVE (POCI-01-0145-FEDER-031823), funded by the FEDER Funds through the COMPETE2020 - POCI, in part by the National Funds (PIDDAC) and in part by the National Science Foundation under Grant No.~ECCS-2029985. The work of P. Sequeira was supported by a Ph.D. Scholarship, grant SFRH/BD/119835/2016 from Fundação para a Ciência e a Tecnologia (FCT), Portugal (POCH program).}}}
\author{Pedro Sequeira\thanks{Faculdade de Engenharia, Universidade do Porto, Portugal (\email{pedro.sequeira@fe.up.pt}, \email{pedro.aguiar@fe.up.pt}).}
	\and A. Pedro Aguiar\footnotemark[2]
	\and João P. Hespanha\thanks{Department of Electrical and Computer Engineering, University of California, Santa Barbara, CA (\email{hespanha@ece.ucsb.edu}).}
}

\headers{Weighted coupled cell networks and invariant synchrony patterns}{Pedro Sequeira, A. Pedro Aguiar, and João P. Hespanha}

\begin{document}
\maketitle

\begin{abstract}
This paper presents a framework based on matrices of monoids for the study of coupled cell networks. We formally prove within the proposed framework, that the set of results about invariant synchrony patterns for unweighted networks also holds for the weighted case.
Moreover, the approach described allows us to reason about any multi edge and multi edge-type network as if it was single edge and single edge-type.
Several examples illustrate the concepts described. Additionally, an improvement of the coarsest invariant refinement (CIR) algorithm to find balanced partitions is presented that exhibits a worst-case complexity of \mbox{$ \mathbf{O}(\vert\mathcal{C}\vert^3) $}, where $\mathcal{C}$ denotes the set of cells.
\end{abstract}

\section{Introduction}
\textbf{Networks} are used to describe systems with multiple components called \textbf{nodes} or \textbf{cells}. These cells can be pairwise connected by \textbf{edges}, describing the effect that one cell has on the other. These edges can be either directed or undirected and can also have weights in order to parameterize their interaction.\\
Such networks are ubiquitous, be it in the natural world or in engineering. From an electronic circuit or the electric grid, to the neural networks in our brain, food webs or the spread of a virus in a pandemic, this is a fundamental structure to study.\\
A big step in understanding these structures was the realization that real-world networks show properties that are pervasive across very different domains of application, such as being `small-world'~\cite{watts1998collective} and having a `scale-free' degree distribution~\cite{barabasi1999emergence}, and the existence of `motifs'~\cite{milo2002network}. All these properties are related to the structure of the network and would not appear in a random one. Reviews on these types of statistical properties and their use in real-world applications are presented in~\cite{newman2003structure}, \cite{boccaletti2006complex}.\\
There are networks in which synchrony between the different cells is of the utmost importance \cite{strogatz1993coupled}. Some examples are the cardiac pacemaker cells responsible for our heartbeat, the flashing of a swarm of fireflies, the consensus problem in control theory and the different gaits in animal locomotion generated by `central pattern generators' (\textbf{CPG}). There are, however, situations in which too much synchronism is actually undesirable, such as in epileptic seizures in the brain.\\
The model most commonly used to describe synchronism is the Kuramoto model, which consists on a large set $ (N\to\infty) $ of simple oscillators that are weakly coupled in an all-to-all fashion. Some reviews on the Kuramoto and its variants can be found in~\cite{arenas2008synchronization}, \cite{dorfler2014synchronization}, \cite{rodrigues2016kuramoto}. This simplistic network structure is, however, in direct opposition to our interest in understanding how the structure of network constrains the function executed by a networked system. The approach in this paper allows us to analyze patterns of synchrony that result exclusively from the topology of the network, regardless of the details of the specific dynamics. We consider both continuous and discrete-time dynamical systems.\\
The theory of coupled cell networks (\textbf{CCN}) was first mathematically formalized in~\cite{stewart2003symmetry}, \cite{golubitsky2005patterns} and \cite{golubitsky2006nonlinear}. In that work, the concept of \textbf{admissibility} was defined by the minimal properties that a function must obey to model a network. This formalism is based on groupoids of bijections between in-neighborhoods of cells. This line of work also introduced the notion of \textbf{quotient network}, which is a smaller network that describes the behavior of the original network when the state of a system is in an \textbf{invariant synchrony pattern}. This means that some cells are sharing the same state and will continue doing so. Some issues arose from the fact that this formalism assumed only single edges between each ordered pair of cells. However, a quotient network might not satisfy this assumption even if the original network of interest does. This issue was solved by the `multiarrow formalism' developed in~\cite{golubitsky2005patterns}, which allows the existence of multiple edges between the same pair of cells and self-loops.
This formalism has been used with simple integer weights, in which the weight is used to represent a number of identical `unitary' edges in parallel. This particular case is the simplest weighted case scenario and the previous formalism happens to be able to cover it. In this paper we properly extend the theory to be able to deal with the general weighted case. This theory is also more general in the sense that it does not require the so called `consistency condition', which is enforced by changing the original network into another one that contains the exact same synchrony patterns. However, this change is not invertible and one loses dynamical information when doing so. We find that such an artificial condition is not necessary.\\
In this paper, a formalism based on matrices of commutative monoids is introduced in \cref{sec:new_formalism}. This formalism allows us to extend the previous known results about CCN's to networks with weighted connections, with arbitrary amount of edges and edge types. We develop the concept of \textbf{oracle functions}, which allows us to evaluate the dynamics of different networks that are composed of the same cell types in a very systematic and self-consistent manner.\\
In \cref{sec:invariant_synchrony} we use the new and more general definitions of admissibility to extend the previous known results about balanced partitions and invariant synchrony spaces.\\
In \cref{sec:vector_output_states} we focus on the particular case where the output set of the admissible functions is a vector space. Furthermore, we provide results in terms of local robustness that apply for this type of spaces.\\
\Cref{sec:quotients} verifies that the connection between quotient networks and synchrony spaces given by balanced partitions work as expected in the general framework.\\
In \cref{sec:lattice} we prove that the lattice structure of balanced partitions is the same as usual. Here, the join operation ($ \vee $) is proved in a novel, algebraic way, instead of the usual duality argument between balanced partitions and invariant subspaces.\\
In \cref{sec:CIR_algorithm_improvement} we propose a novel CIR algorithm for arbitrary weights which has a worst-case time complexity of \mbox{$ \mathbf{O}(\vert\mathcal{C}\vert^3) $} instead of $ \mathbf{O}(\left(\vert\mathcal{E}\vert + \vert\mathcal{C}\vert\right)^4) $ as in~\cite{aldis2008polynomial}, where $\mathcal{C}$ and $\mathcal{E}$ denote the sets of cells and edges, respectively.

\section{Weighted multi-edge CCN's} \label{sec:new_formalism}
In this section we describe a formalism based on a matrix of monoids to represent networks and show that the previous known results about CCN's can be extended to networks with multi edge, multi edge-type, weighted connections.

\subsection{CCN formalism} \label{subsec:CCN_formalism}
We start by introducing the definition of a cell coupled network according to the general weight formalism of this paper. 
\begin{definition}
	A network $ \mathcal{G} $ consists on a set of cells $ \mathcal{C}_{\mathcal{G}} $, where each cell has a type, given by an index set $ \tset=\{1, \ldots, \vert\tset\vert\} $ according to $ \type{G}\colon\ \mathcal{C}_{\mathcal{G}}\to\tset $ and has an \mbox{$ \vert\mathcal{C}_{\mathcal{G}}\vert\times\vert\mathcal{C}_{\mathcal{G}}\vert $} \mbox{in-adjacency} matrix $ M_{\mathcal{G}} $. The entries of $ M_{\mathcal{G}} $ are elements of a family of commutative monoids $ \{\mathcal{M}_{ij}\}_{i,j\in\tset} $ such that $ \left[M_{\mathcal{G}}\right]_{cd} = m_{cd}\in\mathcal{M}_{ij} $, for any cells $ c,d \in \mathcal{C}_{\mathcal{G}} $ with types $ i =\type{G}(c) $, $ j =\type{G}(d) $. 
	\hfill$ \square $
\end{definition}
We will write the monoid operation and the identity element of $ \mathcal{M}_{ij} $ as $ \|_{ij} $ and $ 0_{ij} $ respectively.\\
The entries of $ m_{cd} $ are able to encode the complete connectivity (multi edge, multi edge-type) of the directed edges from $ d $ to $ c $. This is thanks to the algebraic structure of the commutative monoid which we illustrate in the following section.
\begin{remark}
	The subscripts $ _\mathcal{G} $ are omitted when the network of interest is clear from context. 
	\hfill$ \square $
\end{remark}
\subsection{Commutative monoids}
The commutative monoid is the simplest algebraic structure that can be used to describe arbitrary finite parallels of edges. In this paper, we denote the monoid `sum' operation by $ \| $, due to the context in which it is used, with the meaning of `adding in parallel'. Nevertheless, it is convenient to think of this as a sum. Likewise, the notation $ \sum $ is used to describe parallels of multiple edges. In this context, the element $ 0_{\mathcal{M}}\in \mathcal{M} $ should be interpreted as `no edge'.\\
The commutative monoid is associative and commutative. This reflects the fact that, for any given set of edges in parallel, it is irrelevant the order in which we enumerate the individual edges. Those are exactly the properties that provide invariance to this symmetry.\\
Note that we do not impose the existence of inverse elements since it is not guaranteed that we can cancel the effect of a set of edges by adding more edges to it in parallel.\\
We now show how a commutative monoid can be explicitly constructed using what is called a \textbf{presentation}.\\
The first step is to create a \textbf{free commutative monoid}. Given a set $ \wset $, that describes elemental edges, the free commutative monoid on $ \wset $ is $ \mathcal{W} = (\wset^{*},\|_f) $, where $ \wset^{*} $ is the set of all finite multisets of the elements of $ \wset $, which represents all possible finite parallels of edges. Here, $ \|_f $ encodes the multiset sum (free sum) and the element $ 0_{\mathcal{W}} $ is the empty multiset. Note that the set $ \wset $ itself does not need to be finite, or even countable.\\
At this point, the structure is certainly a commutative monoid. However, it is not yet capable of describing an arbitrary one. In particular, it is blind to the possibility of different sets of edges in parallel being equivalent (with regard to the application at hand). For instance, if we are working with the parallel of resistors, we would like to be able to encode into the structure the fact that $ 30\|15 = 20 \| 20 $, from basic circuit theory.\\
In order to generalize this, the second step of the procedure is to quotient the free commutative monoid $ \mathcal{W} $ over a \textbf{congruence relation} $ \mathcal{R} $. A congruence relation on an algebraic structure is an equivalence relation that is compatible with that structure. In our case, this means that we require $ \mathcal{R} $ to be such, that the quotient $ \mathcal{M} = \mathcal{W} / \mathcal{R}$ is a commutative monoid. Here, we think of the equivalence relation $ \mathcal{R} $ as a function in $ \wset^{*} \to \mset $ such that its level sets are the corresponding equivalence classes.\\
In order to satisfy the compatibility condition, we require that if $ \mathcal{R}(a_1) =  \mathcal{R}(a_2) = A $ and $ \mathcal{R}(b_1) = \mathcal{R}(b_2) = B $ then $ \mathcal{R}(a_1 \|_f b_1) = \mathcal{R}(a_2 \|_f b_2) = A\| B $, for any such $ a_1,a_2,b_1,b_2 \in \wset^{*} $. That is, for any equivalence classes, we can choose any of its elements as a representative, and when operating them $ (\|_f) $ the result should be exactly the same, which defines a consistent operation $ \| $ on the equivalence classes.\\
Note that any commutative monoid has a presentation. Given a commutative monoid $ \mathcal{M} =  (\mset,\|) $, we can create the free monoid $ \mathcal{W} = (\mset^{*},\|_f) $. To this end, define the congruence relation $ \mathcal{R} \colon \mset^{*} \to \mset $ such that for any element $ w = w_1 \|_f \ldots \|_f w_k $, with $ w \in \mset^{*} $ and $ w_i\in \mset $, $ i\in \{1,\ldots,k\}$, we have $ \mathcal{R}(w) = w_1\| \ldots \| w_k $. Then, we have that $ \mathcal{M} = \mathcal{W} / \mathcal{R}$.\\
We can also construct our commutative monoid of interest $ \mathcal{M} $ using the set that describes the elemental edges $ \wset $ and defining the congruence relation $ \mathcal{R} $ implicitly using a set of equations $ E $. This can be written as $ \mathcal{M} = \left\langle\wset\vert E\right\rangle $. In the particular case of a free monoid, we write $ \mathcal{M} = \left\langle\wset\vert\right\rangle $. We illustrate these concepts with the following examples.
\begin{exmp}\label{exmp:resistor}
		Consider the commutative monoid generated by finite parallels of resistors. In this case, one has $ \mathcal{M} = \left\langle\wset\vert E\right\rangle $, with 
		\begin{align*}
		\wset = \mathbb{R}^{+}_{0}\cup\{\infty\}
		\end{align*}
		and
		\begin{align*}
		E = 
		\begin{cases}
		w_1 \| w_2 = w_1  w_2/ (w_1 + w_2) & \forall w_1,w_2 \in \wset\colon w_1+w_2 \neq 0\\
		w_1 \| \infty = w_1 & \forall w_1 \in \wset\\
		0 \| 0 = 0
		\end{cases}
		\end{align*}		
	This allows us to verify that indeed $ 30\|15 = 20 \| 20 $. In particular, those parallels are equivalent to an elemental edge of value $ 10 $. For the case of resistors, any set of parallel edges can be simplified into a single edge in $ \wset $. This is not true in general for an arbitrary commutative monoid.\\
	The identity of this monoid is $ 0_{\mathcal{M}} = \infty $. Note that there is no element in $ \mathcal{M} $, except for the identity $ 0_{\mathcal{M}} $ that has an inverse. That is, if there is a finite resistor $ w $ between two nodes, there is no resistor $ w^{-1} $ that we can add in parallel that will cancel it, that is $ w\|w^{-1} = 0_{\mathcal{M}} = \infty $.
	\hfill$ \square $
\end{exmp}
\begin{remark}
	Note that this formalism is extremely general. It allows us to parameterize individual edges with anything we might want, such as complex numbers, vectors, matrices, functions or any data structure as abstract as necessary.
	\hfill$ \square $
\end{remark}
In \cref{exmp:resistor} it can be seen that the zero-valued resistor, which is \textbf{not} the `zero' of the monoid $ (0_\mathcal{M}) $, is an \textbf{annihilator}. That is, an element $ a\in\mathcal{M}  $ such that $ w\|a = a $ for all $ w\in\mathcal{M} $. Not every monoid has an annihilator, but if it exists, it is unique.
\begin{exmp}\label{exmp:primes}
	Consider the commutative monoid $\mathcal{M} = (\mathbb{N},\cdot) $, that is, the integers under the usual product, which has $ 0_{\mathcal{M}} = 1 $. Define now the \textbf{free} monoid $ \mathcal{N} =( \left(\{1\}\cup \mathbb{P}\right)^{*}, \|_f)$, where $ \mathbb{P} $ is the set of prime numbers and $ 0_\mathcal{N} = 1 $. Then, the fundamental theorem of arithmetic says that these monoids are two different ways of describing the exact same object. They are called \textbf{isomorphic}. This means that there is a bijective mapping $ f\colon \left(\{1\}\cup \mathbb{P}\right)^{*}\to \mathbb{N} $ that preserves the monoid structure (isomorphism). In particular, $ f(p_1 \|_f p_2) = f(p_1) \cdot f(p_2) $ for all $ p_1, p_2 \in \left(\{1\}\cup \mathbb{P}\right) $ and $ f(0_{\mathcal{N}}) = 0_{\mathcal{M}} $. We can find such an $ f $ by defining $ f\left(\sum_{i=1}^{k} p_i\right) = \prod_{i=1}^{k} p_i $, in which $ \sum $ is with regard to the multiset sum $ \|_f $. This satisfies $ f(1) = 1 $ and the bijectivity  comes from the uniqueness of prime factorization.
	\hfill$ \square $
\end{exmp}
\begin{remark}
	Note that for the monoid $ \mathcal{N} $ in \cref{exmp:primes}, in opposition to the resistor case (\cref{exmp:resistor}), two elemental edges in parallel are almost never equivalent to another elemental edge. In fact, the only exception is the parallel with identity elements, for which this is inevitable.
	\hfill$ \square $
\end{remark}
\begin{exmp}
	The structure $ \mathcal{M} = (\mathbb{R}\to\mathbb{R},\ast) $, that is, the set of (generalized) functions together with the convolution operation forms a commutative monoid. Its identity is $ 0_{\mathcal{M}} = \delta(\cdot) $, the dirac delta distribution.
	\hfill$ \square $
\end{exmp}
\begin{exmp}\label{exmp:edge_sum_multitype}
	Consider a network with two types of elemental edges, each with its own commutative monoid structure. For instance, $ \mathcal{M}_1 = (\mathbb{R},+) $ and $ \mathcal{M}_2 = ( \mathbb{R}\rightarrow\left[-1,1\right], \cdot)$.\\
	We can merge them into a single commutative monoid by doing a direct product $ \mathcal{M} = \mathcal{M}_1 \times \mathcal{M}_2 $.\\
	An element $ m\in\mathcal{M} $ is an ordered pair $ [m_1, m_2] $ such that $ m_1 \in \mathbb{R} $ and $ m_2 \in \mathbb{R}\rightarrow\left[-1,1\right] $.\\
	The operation $ \| $ of the new monoid is then given by
	\begin{align*}
	w\|v = \left[w_1, w_2\right] \| \left[v_1, v_2\right] = \left[w_1 + v_1, w_2\cdot v_2\right]
	\end{align*}
	That is, the concatenation of applying the respective monoid operations to each component. The identity element of the new monoid is $ 0_{\mathcal{M}} = \left[ 0_{\mathcal{M}_1} ,0_{\mathcal{M}_2}\right] = \left[0,1\right]$.
	\hfill$ \square $
\end{exmp}
This approach of constructing a commutative monoid $ \mathcal{M} $ by merging smaller monoids that represent different edge-types, allows us to use a single monoid structure to fully describe the possible multi edge, multi edge-type connectivity between two cells.\\
Note that for each particular pair of cell types $ i,j \in T $, we could have different monoid structures, which we denote as $ \mathcal{M}_{ij} $, with respect to directed edges from cells of type $ j $ into cells of type $ i $.\\
The network connectivity of the network can then be described by a single matrix whose entries are elements of the appropriate monoid.
\subsection{Partition representations}
In this paper we often refer to each class of a given partition on the set of cells $ \mathcal{C} $ by the term \textbf{color}.\\
Consider a given partition $ \mathcal{A} $ that divides a set of cells $ \mathcal{C} $ into $ r $ colors. We can associate with each color an index from $ \{1,\ldots,r\} $. Then, if a cell $ c\in\mathcal{C} $ has the color associated with index $ k $ we say that $ \mathcal{A}(c) = k $. This association allows us to represent the partition by saying that two cells $ c,d\in\mathcal{C} $ have the same color if and only if $ \mathcal{A}(c) = \mathcal{A}(d) $. We can think of this representation as a column vector or a function.\\
Another very useful representation is to define a partition matrix $ P $ of size $ \vert\mathcal{C}\vert\times r $ such that $ \left[P\right]_{ck} = 1 $ if $ \mathcal{A}(c) = k $ and $ \left[P\right]_{ck} = 0 $ otherwise.\\
Note that the same partition can be indexed by $ \{1,\ldots,r\} $ in different ways which will correspond to multiple partition matrices that are related to each other by a reordering of their columns.\\
The number of colors in a partition is called its \textbf{rank}, which in fact, corresponds to the rank of any of its matrix representations. That is, $ r = rank(\mathcal{A}) = rank(P) $.\\
Given two partitions $ \mathcal{A} $, $ \mathcal{B} $ on a set of cells $ \mathcal{C} $, we say that $ \mathcal{A} $ is \textbf{finer} than $ \mathcal{B} $ if for all $ c,d\in\mathcal{C} $
\begin{equation}
\mathcal{A}(c) = \mathcal{A}(d)
\implies
\mathcal{B}(c) = \mathcal{B}(d)
\label{eq:refinement_def}
\end{equation}
which is denoted as $ \mathcal{A} \leq \mathcal{B}$. Conversely, $ \mathcal{B} $ is said to be \textbf{coarser} than $ \mathcal{A} $. Roughly speaking, \cref{eq:refinement_def} means that if any pair of cells in partition $ \mathcal{A} $ have the same color, then these two cells also have the same color in $ \mathcal{B} $. In other words, if we merge some of the colors of $ \mathcal{A} $ together, we can obtain $ \mathcal{B} $. Conversely, we can obtain $ \mathcal{A} $ by starting with $ \mathcal{B} $ and splitting some of its colors into smaller ones.\\
The \textbf{trivial} partition, in which each individual cell has its own color is the \textbf{finest}, its rank is $ \vert\mathcal{C}\vert  $ and can be represented by any $ \vert\mathcal{C}\vert\times\vert\mathcal{C}\vert $ permutation matrix, one of which is the identity.\\
We will often use the partition and its matrix interchangeably, that is, $ P_{\mathcal{A}} \leq \mathcal{B} $ or $ P_{\mathcal{A}} \leq P_{\mathcal{B}} $ to mean $ \mathcal{A} \leq \mathcal{B}$.\\
If $ \mathcal{A} \leq \mathcal{B}$, then there is a partition $ \mathcal{B}/\mathcal{A} $ on the set of colors of $ \mathcal{A} $ that describes how to merge them in order to achieve partition $ \mathcal{B} $. That is, $ (\mathcal{B}/\mathcal{A}\circ\mathcal{A})(c) = \mathcal{B}(c)$. Equivalently, for partition matrices such that $ P_{\mathcal{A}} \leq P_{\mathcal{B}} $, there exists a partition matrix $ P_{\mathcal{A}\mathcal{B}} $, representing $ \mathcal{B}/\mathcal{A} $ such that $ P_{\mathcal{B}} 
=
P_{\mathcal{A}}  P_{\mathcal{A}\mathcal{B}} $. The next example illustrates these concepts.

\begin{exmp}\label{exmp:pariti_vec_func}
	Consider a set of cells $ \mathcal{C} = \{a,b,c,d\} $ and partitions $ \mathcal{A} = \{\{a,b\},\{c\},\{d\}\} $, $ \mathcal{B} = \{\{a,b,c\},\{d\}\} $. We have that \mbox{$ \mathcal{A}\leq\mathcal{B} $}. Moreover, if one defines the characteristic matrices $ P_{\mathcal{A}},P_{\mathcal{B}} $ as
	\begin{equation*}
	P_{\mathcal{A}} =
	\begin{bmatrix}
	1 & 0 & 0\\
	1 & 0 & 0\\
	0 & 1 & 0\\
	0 & 0 & 1 
	\end{bmatrix}
	\quad
	P_{\mathcal{B}} =
	\begin{bmatrix}
	0 & 1\\
	0 & 1\\
	0 & 1\\
	1 & 0
	\end{bmatrix}
	\end{equation*}
	where the column number of each matrix corresponds to the index of the color that are arbitrarily assigned, then, there is a matrix $ P_{\mathcal{A}\mathcal{B}} $ 
	\begin{equation*}
	P_{\mathcal{A}\mathcal{B}}=
	\begin{bmatrix}
	0 & 1\\
	0 & 1\\
	1 & 0
	\end{bmatrix}
	\end{equation*}	
	such that $ 
	P_{\mathcal{B}} = P_{\mathcal{A}} P_{\mathcal{A}\mathcal{B}}$. Note that in $ P_{\mathcal{A}\mathcal{B}} $, the rows corresponds to the colors in $ \mathcal{A} $, and the columns to the colors in $ \mathcal{B} $, with the $ 1 $'s describing the inclusion relationship between the different colorings.\\ These matrices correspond to a particular indexing such that the partitions can also be represented as the column vectors/functions
	\begin{equation*}
	\mathcal{A}=
	\begin{bmatrix}
	1 \\ 1 \\ 2 \\ 3
	\end{bmatrix}
	\quad
	\mathcal{B}/\mathcal{A} =
	\begin{bmatrix}
	2 \\ 2 \\ 1
	\end{bmatrix}
	\quad
	\mathcal{B}=
	\begin{bmatrix}
	2 \\ 2 \\ 2 \\ 1
	\end{bmatrix}	
	\end{equation*}
	Note for example that \mbox{$ \mathcal{A}(4) = 3 $} and \mbox{$ \mathcal{B}/\mathcal{A}(3) = 1 $}, that is, \mbox{$ (\mathcal{B}/\mathcal{A}\circ\mathcal{A})(4) = 1$} which is equal to $ \mathcal{B}(4) $.
	\hfill$ \square $
\end{exmp}
Note that due to the particular structure of the characteristic matrices it is possible to multiply them together with matrices whose elements are not necessarily in the usual real/complex fields. For a characteristic matrix $ P $ and a given matrix $ M $ of appropriate dimensions, the product $ PM $ is always well-defined as an `expansion' of $ M $ where its rows get replicated. If the product is $ MP $, then it consists of sums of columns of $ M $ which requires a `sum' operation to be defined on its elements (e.g., operations $ \|_{ij} $ if $ M $ represents a network).

\subsection{Admissibility}\label{sec:Gadmissibility}
In this section, we describe the properties that a function $ f\colon\xset\to\yset $ has to respect in order to describe some \textbf{first-order property} of a network. Such a property, when evaluated at a particular cell $ c\in\mathcal{C} $, is given by some $ f_c\colon\xset\to\yset_{i}$, with $ i = \type{}(c) $ that is dependent only on the states of cells in the set $\{ c\cup\mathcal{N}^{-}(c)\}$, where $ \mathcal{N}^{-}(c) $ is the \textbf{in-neighborhood} of $ c $, that is, \mbox{$ \mathcal{N}^{-}(c) = \{d\in\mathcal{C}\colon m_{cd} \neq 0_{ij} , i =\type{}(c) , j =\type{}(d)\} $}.\\
For this purpose, we impose on such functions two minimal assumptions (see \cref{defi:oracle}) that makes them behave as would be expected.\\
Note that these functions can be used to define measurements on a network, e.g., $ \mathbf{y} = f(\mathbf{x}) $, as well as to describe the evolution of dynamical systems, e.g., $ \dot{\mathbf{x}} = f(\mathbf{x}) $ or $ \mathbf{x}^{+} = f(\mathbf{x}) $.\\
This does not mean that everything on a network has to (or can) be defined by such a function. For instance, the second derivative or the two-step evolution of the mentioned dynamical systems will not be of this form. Those functions will be `second-order' in the sense that they are dependent on they first and second in-neighborhoods (neighbor of neighbor). They are, however, fully defined from the original first-order functions.
\\
Consider the simple network of \cref{fig:edge_merginga}, (which could be part of a larger network) consisting on cell $ 3 $ and its \mbox{in-neighborhood}. We have cell types $ \tset =\{1,2\} $ which represent `circle' and `square' cells, respectively. We use $ \| $ to mean $ \|_{12} $ since it is the only monoid operation in this example. Since cells $ 1 $ and $ 2 $ are of the same cell type (square) ($ \type{}(1) = \type{}(2) = 2 $) and also, are currently in the same state ($ x_1=x_2 $), the total input received by cell $ 3 $ at that instant, is the same as if both edges originated from a single `square' cell with that state.
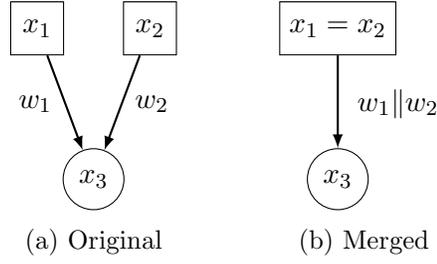
\begin{figure}[h]
	\centering
	\begin{subfigure}[t]{0.23\textwidth}
		\centering
		\input{edge_merginga.tikz}
		\caption{Original}
		\label{fig:edge_merginga}
	\end{subfigure}
	\begin{subfigure}[t]{0.23\textwidth}
		\centering
		\input{edge_mergingb.tikz}
		\caption{Merged}
		\label{fig:edge_mergingb}
	\end{subfigure}
	\centering
	\caption{Edge merging}
	\label{fig:edge_merging}
\end{figure}
Furthermore, if the weights cancel \mbox{$ (w_1 \| w_2 = 0_{12}) $}, then cell $ 3 $ should act as if cells $ 1 $, $ 2 $ are not there whenever $ x_1 = x_2 $.
Moreover, the input received by a cell is independent of how we draw the network, that is, we do not expect $ f_3 $ to be different if cell $ 2 $ was at the left of cell $ 1 $. That is, $ f_3 $ should obey
\begin{align*}
f_3\left(x_3; \begin{bmatrix}w_1 & w_2\end{bmatrix},
\begin{bmatrix}x_1 \\ x_2\end{bmatrix}\right) 
=
f_3\left(x_3; \begin{bmatrix}w_2 & w_1\end{bmatrix},
\begin{bmatrix}x_2 \\ x_1\end{bmatrix}\right) 
\end{align*}
Extending these principles, if two cells $ c,d $ are in the same internal state $ x_c = x_d = x $ and if the edge-compression argument can transform their \mbox{in-neighborhoods} into the same network, then $ f_c = f_d$. This is illustrated in \cref{fig:input_equiv}, with the monoid operation $ \| $ being the usual addition.
\begin{figure}[h]
	\centering
	\begin{subfigure}[t]{0.30\textwidth}
		\centering
		\input{input_equiv1.tikz}
		\caption{First input set}
		\label{fig:input_equiv1}
	\end{subfigure}
	\begin{subfigure}[t]{0.30\textwidth}
		\centering
		\input{input_equiv2.tikz}
		\caption{Second input set}
		\label{fig:input_equiv2}
	\end{subfigure}
	\begin{subfigure}[t]{0.30\textwidth}
		\centering
		\input{input_equiv_compressed.tikz}
		\caption{Common edge compression to both input sets}
		\label{fig:input_equiv_compressed}
	\end{subfigure}
	\centering
	\caption{Input equivalent networks}
	\label{fig:input_equiv}
\end{figure}
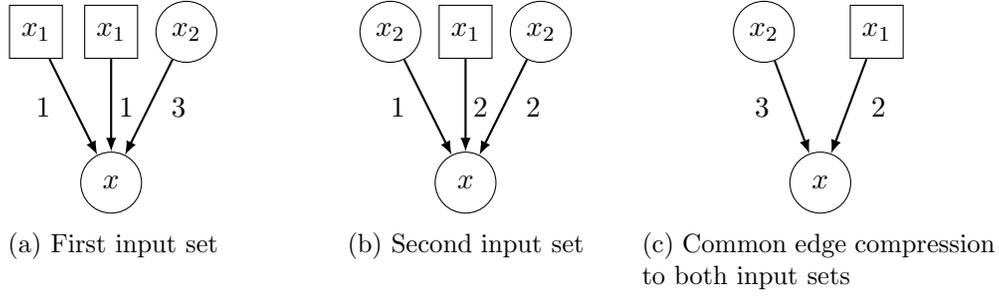
If we write the state and weight vectors from left to right, we get
\begin{align*}
\begin{split}
\mathbf{x}_\mathcal{A} &= \begin{bmatrix} x_1 & x_1 & x_2\end{bmatrix}^{\top}\\
\mathbf{x}_\mathcal{B} &= \begin{bmatrix} x_2 & x_1 & x_2\end{bmatrix}^{\top}\\
\overline{\mathbf{x}} &= \begin{bmatrix} x_2 & x_1\end{bmatrix}^{\top}
\end{split}
\begin{split}
\mathbf{w}_\mathcal{A} &= \begin{bmatrix} 1 & 1 & 3\end{bmatrix}\\
\mathbf{w}_\mathcal{B} &= \begin{bmatrix} 1 & 2 & 2\end{bmatrix}\\		
\overline{\mathbf{w}} &= \begin{bmatrix} 3 & 2\end{bmatrix}
\end{split}
\end{align*} 
for each of the \cref{fig:input_equiv1,fig:input_equiv2,fig:input_equiv_compressed}, respectively. The process of \mbox{edge-merging} the in-neighborhoods from \cref{fig:input_equiv1,fig:input_equiv2} into \cref{fig:input_equiv_compressed} can be described through the partition matrices $ P_{\mathcal{A}},P_{\mathcal{B}} $
\begin{align*}
P_{\mathcal{A}} =
\begin{bmatrix}
0 & 1 \\ 0 & 1 \\ 1 & 0 
\end{bmatrix}
\quad
P_{\mathcal{B}} =
\begin{bmatrix}
1 & 0 \\ 0 & 1 \\ 1 & 0 
\end{bmatrix}
\end{align*}
such that
\begin{align*}
\mathbf{x}_\mathcal{A} = P_{\mathcal{A}}\overline{\mathbf{x}},
\quad
\mathbf{x}_\mathcal{B} = P_{\mathcal{B}}\overline{\mathbf{x}}\\
\mathbf{w}_\mathcal{A}P_{\mathcal{A}} = \overline{\mathbf{w}} = \mathbf{w}_\mathcal{B}P_{\mathcal{B}}
\end{align*}
where we considered that the edge merging can only be done with edges that come from cells of the same type in the same state. That is, we have also assumed implicitly that
\begin{align*}
\type{A} = P_{\mathcal{A}}\overline{\type{}},
\quad
\type{B} = P_{\mathcal{B}}\overline{\type{}}
\end{align*}
where $ \type{A} $, $ \type{B} $ and $ \overline{\type{}} $ describe the typings of the corresponding cells, that is,
\begin{align*}
\type{A} = \begin{bmatrix} 2 & 2 & 1\end{bmatrix}^{\top} \qquad
\type{B} = \begin{bmatrix} 1 & 2 & 1\end{bmatrix}^{\top} \qquad
\overline{\type{}} = \begin{bmatrix} 1 & 2\end{bmatrix}^{\top}
\end{align*}
This type checking can be omitted if declare that, by definition, the states of two cells can only be compared in the first place if they are of the same type. That is, if $ P $ satisfies this assumption in $ \mathbf{x} = P\overline{\mathbf{x}} $, there will be no danger of trying to `sum' elements of different monoids in $ \mathbf{w} P = \overline{\mathbf{w}} $.\\
This edge-merging example motivates the following definitions.
\begin{definition} \label{defi:oracle}
	Consider the set of cell types $ \tset $, and the related sets $ \{\xset_j, \yset_j\}_{j\in\tset} $ together with a family of commutative monoids $ \{\mathcal{M}_{ij}\}_{j\in\tset} $, for a given fixed $ i\in\tset $. Take a function $ \hat{f}_i $ defined on
	\begin{align}
	\hat{f}_i\colon \xset_i\times\bigcup^\circ_{\mathbf{k}\in\mathbb{N}_0^{\vert T\vert} } \left( \mathcal{M}_i^{\mathbf{k}} \times \xset^{\mathbf{k}} \right)  \to \yset_i
	\end{align}
	where $ \bigcup\limits^\circ $ denotes the disjoint union and $ \mathbf{k} = \left[k_1, \ldots, k_{\vert\tset\vert}\right]$ is a multi-index with $ k_i \in \mathbb{N}_0 $ for all $ i\in\tset$ such that $ \xset^{\mathbf{k}} := \xset_1^{k_1} \times \ldots \times \xset_{\vert\tset\vert}^{k_{\vert\tset\vert}} $ and $ \mathcal{M}_i^{\mathbf{k}} := \mathcal{M}_{i1}^{k_1} \times \ldots \times \mathcal{M}_{i\vert\tset\vert}^{k_{\vert\tset\vert}} $.\\
	The function $ \hat{f}_i $ is called an \textbf{oracle component of type i}, if it has the property that for every $ x\in\xset_i $, $ \mathbf{x} \in \xset^{\mathbf{k}} $, $ \mathbf{w} \in \mathcal{M}_i^{\mathbf{k}} $, $ \overline{\mathbf{x}}\in\xset^{\overline{\mathbf{k}}} $, $ \overline{\mathbf{w}}\in\mathcal{M}_i^{\overline{\mathbf{k}}} $ and partition matrix $ P $, that satisfy
	\begin{equation}
	\begin{cases}
	\mathbf{x} = P\overline{\mathbf{x}}\\
	\mathbf{w}P = \overline{\mathbf{w}}
	\end{cases}
	\label{eq:condi_for_equal_admi}
	\end{equation}
	then
	\begin{equation}
	\hat{f}_i(x;\mathbf{w},\mathbf{x}) = 
	\hat{f}_i(x;\overline{\mathbf{w}},\overline{\mathbf{x}})
	\label{eq:admi_equality}
	\end{equation}
	Furthermore, if $ \mathbf{w} $ has its $ k^{th} $ element (corresponding to cell $ c_k $) equal to $ 0_{ij} $, with $ j = \type{}(c_k) $, then 
	\begin{equation}
	\hat{f}_i(x;\mathbf{w},\mathbf{x}) = 
	\hat{f}_i(x;\mathbf{w}_{-k},\mathbf{x}_{-k})
	\label{eq:admi_equality_0_equal}
	\end{equation}
	where $ \mathbf{w}_{-k} $, $ \mathbf{x}_{-k} $ denotes the result of removing the $ k^{th} $ element of the original vectors $ \mathbf{w} $, $ \mathbf{x} $.\\
	In \cref{eq:condi_for_equal_admi}, equality between states assumes compatible cell types.
	\hfill$ \square $
\end{definition}
The disjoint union allows us to distinguish neighborhoods of different types, that is, the set $ \xset_1\times\xset_1 $ is always taken as a different set from $ \xset_1\times\xset_2 $ even in the particular case of $ \xset_1 = \xset_2 $.
\begin{remark}
	Note that the way the domain of $ \hat{f}_i $ was defined allows us to deal with variable input set configurations. This is an equivalent, but cleaner way of defining a family of functions, each on a different domain based on its specific input set, such that they are all connected by the self-consistency rules that we expect from them. This way, we can use a single function to describe what really matters to us, that is, describing how a cell is affected by its in-neighbors.
	\hfill$ \square $
\end{remark}
\begin{remark}
	It is easy to verify that for a \textbf{permutation} matrix $ P $ that preserves cell typing $ ( \xset_{\mathbf{k}}  = P\xset_{\mathbf{k}} ) $ we have
	\begin{equation*}
	\hat{f}_i(x;\mathbf{w},\mathbf{x}) = \hat{f}_i(x;\mathbf{w}P^{\top},P\mathbf{x})
	\end{equation*}
	which is in accordance to the idea that a cell does not care about the order in which its input set is drawn.
	\hfill$ \square $
\end{remark}
The \textbf{oracle set} is the set of all $ \vert T\vert $-tuples of oracle components, such that each element of the tuple represents one of the types in $ T $. It is denoted as
\begin{align*}
\hat{\mathcal{F}}_{T} = \prod_{i=1}^{\vert T\vert}\hat{\mathcal{F}}_{i} 
\end{align*}
where $ \hat{\mathcal{F}}_{i} $ is the set of all oracle components of type $ i $ . We are always implicitly assuming sets $ \{\xset_i, \yset_i\}_{i\in\tset} $ and commutative monoids $ \{\mathcal{M}_{ij}\}_{i,j\in\tset} $.
Note that modeling some aspect of a network that follows our assumptions is effectively choosing one of the elements of $ \hat{\mathcal{F}}_{T} $, which we call \textbf{oracle functions}.
\begin{definition}\label{defi:F_G_admissibility}
	Consider a network $ \mathcal{G} $ defined on a cell set $ \mathcal{C} $ with cell types in $ \tset $ according to the cell type partition $ \type{} $, and an in-adjacency matrix $ M $. Assume without loss of generality that the cells are ordered according to the cell types such that we can associate with the network a state $ \xset := \xset^{\mathbf{k}} $ and output $ \yset := \yset^{\mathbf{k}} $ sets, with $ \vert \mathcal{C} \vert = \vert \mathbf{k} \vert $
	\begin{align*}
	\vert \mathbf{k} \vert = \sum_{i = 1}^{\vert T \vert} k_i
	\end{align*}
	and $ k_i $ is the number of cells in $ \mathcal{C} $ of type $ i\in\tset $.\\
	A function \mbox{$ f\colon\xset\to\yset $}, given as
	\begin{align*}
	f = (f_c)_{c\in\mathcal{C}},\quad \text{with } f_c\colon\xset\to\yset_{i}, \quad i = \type{}(c)
	\end{align*}
	is said to be \mbox{$ \mathcal{G} $\textbf{-admissible}} if there is some oracle function \mbox{$ \hat{f}\in\hat{\mathcal{F}}_{T}$, $ \hat{f} =  (\hat{f}_i)_{i\in\tset}$} such that
	\begin{equation}
	f_c(\mathbf{x}) = \hat{f}_{i}\left(x_c; \mathbf{m}_c, \mathbf{x} \right)\label{eq:dynamics_dependence}
	\end{equation}
	for $ \mathbf{x}\in\xset $, where $ x_c $ is the $ c^{th} $ coordinate of $ \mathbf{x} $ and $ \mathbf{m}_c $ is the $ c^{th} $ row of matrix $ M $. In this case we write $ f = \hat{f}\rvert_{\mathcal{G}} $.
	\hfill$ \square $
\end{definition}
The set of all \mbox{$ \mathcal{G} $-admissible} functions is denoted as $ \mathcal{F}_{\mathcal{G}} $. It can be thought of as the result of evaluating $ \hat{\mathcal{F}}_{T} $ at $ \mathcal{G} $, which can be written as $ \hat{\mathcal{F}}_{T}\rvert_{\mathcal{G}} $.
Note that process of evaluating oracle functions at a network is not necessarily injective. There might be oracle functions \mbox{$ \hat{f},\hat{g}\in\hat{\mathcal{F}}_{T} $} with $ \hat{f} \neq \hat{g} $ such that $ \hat{f}\rvert_{\mathcal{G}} = \hat{g}\rvert_{\mathcal{G}} $.\\
The next example makes explicit the relation between the connectivity graph of a network and how that constraints any possible admissible function that acts on it.
\begin{exmp}\label{exmp:admiss_funcs}	
	\Cref{fig:Gadmissibility} shows an example of a \textbf{CCN} of three cells of the same type.
	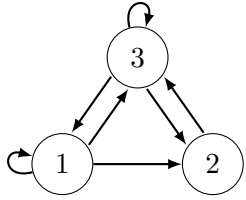
\begin{figure}[h]
		\centering
		\input{Gadmissibility_example.tikz}
		\caption{Simple network with admissible functions that have the structure given by \cref{eq:func_match1,eq:func_match2,eq:func_match3}}
		\label{fig:Gadmissibility}
	\end{figure}	
	This \textbf{CCN} can be described by the in-adjacency matrix $ M $ 
	\begin{equation} M = 
	\begin{bmatrix}
	1 & 0 & 1\\
	1 & 0 & 1\\
	1 & 1 & 1
	\end{bmatrix}
	\end{equation}
	together with the cell type partition \mbox{$ \type{} = \{\{1,2,3\}\} $.} This means that a suitable $ f\in\mathcal{F}_{\mathcal{G}} $ should have the following structure
	\begin{align}
	f_1(\mathbf{x}) &= \hat{f}(x_1;\begin{bmatrix} 1 & 0 & 1\end{bmatrix},\mathbf{x})\label{eq:func_match1}\\
	f_2(\mathbf{x}) &= \hat{f}(x_2;\begin{bmatrix} 1 & 0 & 1\end{bmatrix},\mathbf{x})\label{eq:func_match2}\\
	f_3(\mathbf{x}) &= \hat{f}(x_3;\begin{bmatrix} 1 & 1 & 1 \end{bmatrix},\mathbf{x})\label{eq:func_match3}
	\end{align}
	for some $ \hat{f}\in\hat{\mathcal{F}}_{T}  $. Note that here $ T $ only has one type, that is $ \hat{\mathcal{F}}_{\circ} = \hat{\mathcal{F}}_{T} $. Therefore, we use $ \hat{f} $ interchangeably as both oracle function and oracle component. It is exactly the same treatment as not differentiating between an one-dimensional vector and the element it contains.
	\hfill$ \square $
\end{exmp}
To make more explicit the importance of a rigorous definition for admissibility, the following example presents a case that might look reasonable at a first glance but ends up not being admissible.
\begin{exmp}
	Consider the simple network in \cref{fig:edge_merging} that was used to illustrate the edge merging concept. We will propose a function on the original network \cref{fig:edge_merginga} and verify if it satisfies our assumptions.\\
	We consider that the cells have associated state and output sets given by $ \xset_1 = \xset_2 = \yset_1 = \mathbb{R} $, such that $ \tset = \{1,2\} $ identify the cell types `circle' and `square' respectively.\\
	The directed edges from `square' into `circle' are in $ \mathcal{M}_{12} $ and we denote $ \|_{12} $ by $ \| $ for simplicity.\\
	Given functions $ g\colon \mathbb{R}\to\mathbb{R}$ and \mbox{$ p\colon \mathcal{M}_{12} \to \mathbb{R} $}, with \mbox{$ p(0_{12}) = 0 $}, it is tempting to think that the function $ f_3 $, could be modeled by
	\begin{align}
	f_3(\mathbf{x}) = g(x_3) + p(w_1)x_1 + p(w_2) x_2 + p(w_1) p(w_2) x_1 x_2 
	\label{eq:non_admissible_examp}
	\end{align}
	After all, if we simultaneously switch $ w_1 \leftrightarrow w_2 $ and $ x_1 \leftrightarrow x_2 $, $ f_3 $ would still look the same. Consider, $ w_1 = w $, $ x_1 = x $ and $ w_2 = 0_{12} $. Then, if cell $ 3 $ only had one neighbor (of type square), $ f_3 $ would be given by
	\begin{align*}
	f_3(\mathbf{x}) = g(x_3) + p(w) x
	\end{align*}
	If we have $ x_1 = x_2 = x_{12}$, from the edge-merging principle, we should be in the situation of \cref{fig:edge_mergingb}. We would have
	\begin{align*}
	f_3(\mathbf{x}) = g(x_3) + p(w_1 \| w_2) x_{12}
	\end{align*}
	However, from direct substitution on \cref{eq:non_admissible_examp} we obtain 
	\begin{align*}
	f_3(\mathbf{x}) = g(x_3) + ( p(w_1) + p(w_2) + p(w_1) p(w_2) x_{12} )x_{12} 
	\end{align*}
	which means that this is not admissible since 
	\begin{align*}
	p(w_1 \| w_2) =  p(w_1) + p(w_2) + p(w_1) p(w_2) x_{12} 
	\end{align*}
	goes against the assumption that $ p $ depends only on the edge weights.\\
	Consider that instead $ f_3 $ was modeled as
	\begin{align*}
	f_3(\mathbf{x}) = g(x_3) + p(w_1) x_1 + p(w_2) x_2 + p(w_1) p(w_2) \frac{x_1 + x_2 }{2}
	\label{eq:corrected_admissible_examp1}
	\end{align*}
	Following the exact same approach this requires
	\begin{align*}
	p(w_1 \| w_2) =  p(w_1) + p(w_2) + p(w_1) p(w_2) 
	\end{align*}
	which is a valid constraint. It only depends on its inputs and is compatible with a commutative monoid structure, that is,
	\begin{align*}
	p(w \| 0_{12}) &= p(w)\\
	p(w_1 \| w_2) &= p(w_2 \| w_1)\\
	p((w_1 \| w_2) \| w_3) &= p(w_1 \| (w_2 \| w_3))
	\end{align*}
	Note that for each of the three equalities,  the inputs for both members are the same element of $ \mathcal{M}_{12} $. The same input of a function has to output the same value. 
	\hfill$ \square $
\end{exmp}
\begin{remark}
	Assume that in the previous example there was an annihilator element in $ \mathcal{M}_{12} $, that is, an element $ a_{12}\in\mathcal{M}_{12} $ such that $ w\| a_{12} = a_{12} $ for all $ w\in\mathcal{M}_{12}$. Then, either $ p(a_{12}) = -1 $ or we are in the non interesting case where $ p(w) = 0 $ for all $ w\in\mathcal{M}_{12} $.\\
	An example of an annihilator is the short-circuit $ (R=0) $ with regard to the parallel of resistors, as in \cref{exmp:resistor}.
	\hfill$ \square $
\end{remark}
\begin{remark}
	Note that we solved the problem of the $ 2 $-coupling component being quadratic on $ x_{12} $ when $ x_1 = x_2 = x_{12} $ by modeling that component additively with $ (x_1+x_2)/2 $. This is not the only approach. Consider now that we have some function $ q\colon \mathcal{M}_{12} \times\mathbb{R} \to \mathbb{R} $ with $ q(0_{12},x) = 0 $ such that we model $ f_3 $ by
	\begin{align*}
	f_3(\mathbf{x}) = g(x_3) + q(w_1,x_1) + q(w_2, x_2) + q(w_1,x_1) q(w_2,x_2)
	\end{align*}
	This is also valid when
	\begin{align*}
	q(w_1 \| w_2,x) =  q(w_1,x) + q(w_2,x) + q(w_1,x) q(w_2,x) 
	\end{align*}
	Similarly, this also compatible with a commutative monoid structure.
	\hfill$ \square $	
\end{remark}

\section{Invariant synchrony}\label{sec:invariant_synchrony}
In \cref{sec:Gadmissibility} we introduced the concept of oracle components $ \hat{f}_i\in\hat{\mathcal{F}}_{i} $ which describe the way a cell of type $ i\in\tset $ reacts to its input set in a way that is self-consistent. That is, having the same state and equivalent inputs generates the same output and the output is only functionally dependent on the in-neighbors.\\
The functional modeling of a network $ \mathcal{G} $ should then be though of as choosing an oracle function $ \hat{f}\in\hat{\mathcal{F}}_{T} $ and constraining it to $ \mathcal{G} $ as in \cref{defi:F_G_admissibility}. This gives us an $ \mathcal{G} $-admissible function \mbox{$ f\in\mathcal{F}_{\mathcal{G}} $} such that $ f = \hat{f}\rvert_{\mathcal{G}} $.\\
The oracle components, which are the core concept of the modeling of a network were defined in terms of equality relationships. This motivates the study of whether a set of equality relations in the state set $ \xset $ are preserved by $ f $ regardless of the specific admissible $ f $ that acts on the network.
\begin{definition}
	Given a partition $ \mathcal{A}\leq\type{} $ we call the subset of $ \xset $
	\begin{equation}
	\Delta_{\mathcal{A}}^{\xset} 
	=
	\{ \mathbf{x}\in\xset\colon
	\mathcal{A}(c)
	=
	\mathcal{A}(d)
	\implies x_c = x_d\}
	\end{equation}
	the \textbf{polydiagonal} of $ \mathcal{A} $ in $ \xset $.
	\hfill$ \square $
\end{definition}
This means that any $ \mathbf{x}\in \Delta_{\mathcal{A}}^{\xset} $ can be given by $ \mathbf{x} = P\overline{\mathbf{x}} $ for some $ \overline{\mathbf{x}} $, where $ P $ is a characteristic matrix of $ \mathcal{A} $. Note that
\begin{equation}
\mathcal{A} \leq \mathcal{B} \Longleftrightarrow \Delta_{\mathcal{A}}^{\xset} \supseteq \Delta_{\mathcal{B}}^{\xset}
\end{equation}
\begin{definition}
	If for a function $ f\colon\xset\to\yset $ and a partition $ \mathcal{A} \leq \type{} $ we have
	\begin{equation}
	f\left(\Delta_{\mathcal{A}}^{\xset}\right) \subseteq \Delta_{\mathcal{A}}^{\yset}
	\label{eq:polysynchronous}
	\end{equation}
	then $ f $ is $ \mathcal{A} $-\textbf{invariant}.
	\hfill$ \square $
\end{definition}
Note that if $ f $ is $ \mathcal{A} $-invariant, then for every $ \mathbf{x} \in\xset  $ such that $ \mathbf{x} = P\overline{\mathbf{x}} $, with $ P $ representing $ \mathcal{A} $, there is $ \overline{\mathbf{y}} $ such that \mbox{$ f(P\overline{\mathbf{x}}) = P\overline{\mathbf{y}} $}. This means that there is a function $ \overline{f}\colon\overline{\xset}\to\overline{\yset} $ with $ \xset = P\overline{\xset} $, $ \yset = P\overline{\yset} $ such that
\begin{equation}
f(P\overline{\mathbf{x}})
=
P\overline{f}(\overline{\mathbf{x}})
\label{eq:f_in_invariant_space}
\end{equation}
Consider a discrete-time system $ \mathbf{x}^{+} = f(\mathbf{x}) $ that evolves according to an $ \mathcal{G} $-admissible map $ f\colon\xset\to\xset $. If $ f $ is \mbox{$ \mathcal{A} $-invariant}, then
\begin{equation}
\mathbf{x}_{n_0} \in \Delta_{\mathcal{A}}^{\xset} \implies \mathbf{x}_n \in \Delta_{\mathcal{A}}^{\xset} \quad \forall n\in\mathbb{N} \colon n\ge n_0 
\end{equation}
Similarly, for a continuous-time system \mbox{$ \dot{\mathbf{x}} = f(\mathbf{x}) $} that evolves according to an \mbox{$ \mathcal{G} $-admissible} vector field \mbox{$ f(\mathbf{x})\colon\xset\to T_{\mathbf{x}}\xset $} where $ f $ is Lipschitz, $ \xset $ is a smooth manifold and $ T_{\mathbf{x}}\xset $ its tangent space at $ \mathbf{x} $. If $ f $ is $ \mathcal{A} $-invariant, then
\begin{equation}
\mathbf{x}(t_0) \in \Delta_{\mathcal{A}}^{\xset} \implies \mathbf{x}(t) \in \Delta_{\mathcal{A}}^{\xset} \quad \forall t\in\mathbb{R}
\end{equation}
Note that in both cases the polydiagonals $ \Delta_{\mathcal{A}}^{\xset} $ are invariant with respect to the dynamics. Moreover, the evolution of $ \mathbf{x} $ is fully determined by $ \overline{\mathbf{x}} $, which in turn evolves according to
\begin{equation}
\overline{\mathbf{x}}^{+}\backslash\dot{\overline{\mathbf{x}}}
=
\overline{f}(\overline{\mathbf{x}})
\end{equation}
The following concept relates the structure of a network with its capability to preserve a given synchrony pattern. This works by verifying whether, for a given coloring of the network, if cells that have the same color also have colored input sets that are equivalent.
\begin{definition}
	Consider a network $ \mathcal{G} $ defined on a cell set $ \mathcal{C} $ with a cell type partition $ \type{} $ and an in-adjacency matrix $ M $. A partition $ \mathcal{A}\leq \type{}$ with characteristic matrix $ P $ is said to be \textbf{balanced} on $ \mathcal{G} $ if for all $ c,d\in\mathcal{C} $
	\begin{align}\label{eq:balanced_cond}
	\mathcal{A}(c)
	=
	\mathcal{A}(d)
	\implies
	\mathbf{m}_c P = \mathbf{m}_d P
	\end{align}
	where $ \mathbf{m}_c, \mathbf{m}_d $ are the $ c^{th} $ and $ d^{th} $ rows of matrix $ M $.
	\hfill$ \square $
\end{definition}
A balanced partition is usually indicated with the symbol $ \bp $. 
Note that a partition is balanced if and only if there is a matrix $ Q $ of elements in the appropriate monoids $ \{\mathcal{M}_{ij}\}_{i,j\in\tset} $ such that
\begin{equation}\label{eq:balanced_cond_matrix_form}
M P = P Q
\end{equation}
Clearly, the trivial partition is always balanced. That is, for any $ M $, the condition \cref{eq:balanced_cond_matrix_form} is satisfied with $ P = I $ and \mbox{$ Q = M $}. We now show the following result.
\begin{theorem}\label{thm:balanced_implies_all}
	Consider $ \mathcal{F}_{\mathcal{G}} $, defined on a network $ \mathcal{G} $ with sets $ \{\xset_i, \yset_i\}_{i\in\tset} $ and commutative monoids $ \{\mathcal{M}_{ij}\}_{i,j\in\tset} $.\\
	If a partition $ \bp $ on $ \mathcal{G} $ is balanced, then every $ f\in\mathcal{F}_{\mathcal{G}} $ is $ \bp $-invariant.
	\hfill$ \square $
\end{theorem}
\begin{proof}
	Consider $ \mathbf{x}\in\Delta_{\bp}^{\xset} $, that is, $ \mathbf{x} = P\overline{\mathbf{x}} $ for some $ \overline{\mathbf{x}} $. From the balanced definition \cref{eq:balanced_cond} we have \mbox{$ \bp(c) = \bp(d) \implies \mathbf{m}_c P = \mathbf{m}_d P $}, which satisfies conditions \cref{eq:condi_for_equal_admi} and therefore from the definition of admissibility $ f_c(P\overline{\mathbf{x}}) = f_d(P\overline{\mathbf{x}}) $. This means that there is a $ \overline{f} $ such that $ f(P\overline{\mathbf{x}}) = P\overline{f}(\overline{\mathbf{x}}) $ and every $ \mathcal{G} $-admissible function $ f\in \mathcal{F}_{\mathcal{G}} $ is \mbox{$ \bp $-invariant}.
\end{proof}
The unweighted version of this theorem is presented in~\cite{stewart2003symmetry}, \cite{golubitsky2005patterns} and \cite{golubitsky2006nonlinear} as an equivalence. Note however, that although the forward direction is a strong result ($ \bp $ balanced $ \implies f $ is $ \bp $-invariant $\forall f\in\mathcal{F}_{\mathcal{G}} $, for any $ \mathcal{F}_{\mathcal{G}} $), the backwards one is extremely weak by comparison, requiring every single $ \mathcal{G} $-admissible function to be $ \bp $-invariant. For this reason, we present the backwards direction separately.
\begin{theorem}\label{thm:all_implies_balanced}
	Consider $ \mathcal{F}_{\mathcal{G}} $, defined on a network $ \mathcal{G} $ with sets $ \{\xset_i, \yset_i\}_{i\in\tset} $ and commutative monoids $ \{\mathcal{M}_{ij}\}_{i,j\in\tset} $.\\
	For a partition $ \bp $ on $ \mathcal{G} $, if every $ f\in\mathcal{F}_{\mathcal{G}} $ is $ \bp $-invariant, then $ \bp $ is balanced.
	\hfill$ \square $
\end{theorem}
\begin{proof}
	Assume $ \bp $ is not balanced. Then, there are \mbox{$ \bp(c) = \bp(d) $} such that $ \mathbf{m}_c P \neq \mathbf{m}_d P $. Consider $ k $ one of the colors in which they differ. That is, $ \left[\mathbf{m}_c P\right]_k \neq \left[\mathbf{m}_d P\right]_k $.\\ Moreover, choose some state $ \mathbf{x}\in\Delta_{\bp}^{\xset} $, that is, $ \mathbf{x} = P\overline{\mathbf{x}} $, such that $ \overline{x}_k $ is different from all other entries of $ \overline{\mathbf{x}} $.\\
	Then, there is an $ f\in\mathcal{F}_{\mathcal{G}} $, that is, $ f = \hat{f}\rvert_{\mathcal{G}} $ such that $ \hat{f}_{i}\left(x; \mathbf{w}, \mathbf{x} \right) = y_1 $ if $\mathbf{w}$, summed over the entries such that $ \mathbf{x} $ is $ \overline{x}_k $, results into $ \left[\mathbf{m}_c P\right]_k $, and $ \hat{f}_{i}\left(x; \mathbf{w}, \mathbf{x} \right) = y_2 $ otherwise, with $ y_1\neq y_2 $, $ y_1, y_2 \in \yset_{i} $ and $ i = \type{}(c)$.\\
	Then, we have an $ f\in\mathcal{F}_{\mathcal{G}} $ and $ \overline{\mathbf{x}} $ such that $ f\left(P\overline{\mathbf{x}}\right) \not\in \Delta_{\bp}^{\yset} $. Therefore, the only way to not be able to find such an $ f $, is for $ \bp $ to be balanced.
\end{proof}
For the unweighted case, the authors of~\cite{stewart2003symmetry}, \cite{golubitsky2005patterns} and \cite{golubitsky2006nonlinear} proved the backwards direction by actually proving a stronger result that they hid away in the proof.\\
Their result is stronger in the sense that they show that only a particular subset of $ \mathcal{F}_\mathcal{G} $ is necessary to be $ \bp $-invariant to enforce $ \bp $ to be balanced. This works in the weaker framework of the unweighted case, which corresponds to the simple monoid structure $ \mathcal{M} = (\mathbb{N}_0,+) $.\\
Nevertheless, a stronger result deserves to be shown in its own right. Therefore, in the next section, we present similar results for special cases of monoids that extends their previous result.
\section{Output vector spaces}\label{sec:vector_output_states}
The following results apply when the output sets are vector spaces.
\begin{lemma}\label{lemma:admissible_vector_space}
	Consider $ \hat{\mathcal{F}}_{i} $, defined on an output set $ \yset_i $ that is a vector space. Then, $ \hat{\mathcal{F}}_{i} $ is itself a vector space.
	\hfill$ \square $
\end{lemma}
\begin{proof}
	Consider components $ \hat{f}_i,\hat{g}_i\in \hat{\mathcal{F}}_{i}$. Defining, \mbox{$ \hat{h}_i = \hat{f}_i + \hat{g}_i $}, means that if conditions \cref{eq:condi_for_equal_admi} are satisfied for $ \hat{f}_i $ and $ \hat{g}_i $, then
	\begin{align*}
	\hat{h}_i(x;\mathbf{w},\mathbf{x}) &= \hat{f}_i(x;\mathbf{w},\mathbf{x}) + \hat{g}_i(x;\mathbf{w},\mathbf{x})\\
	&= \hat{f}_i(x;\overline{\mathbf{w}},\overline{\mathbf{x}}) + \hat{g}_i(x;\overline{\mathbf{w}},\overline{\mathbf{x}})\\
	&= \hat{h}_i(x;\overline{\mathbf{w}},\overline{\mathbf{x}})
	\end{align*} 
	which means that $ \hat{h}_i $ satisfies condition~\cref{eq:admi_equality}. Condition~\cref{eq:admi_equality_0_equal} is verified in exactly the same way.	
	Therefore, $ \hat{h}_i\in\hat{\mathcal{F}}_{i}$.\\
	If $ \hat{f}_i $ satisfies \cref{eq:admi_equality,eq:admi_equality_0_equal} then $ \alpha\hat{f}_i $ also satisfies it for any scalar $ \alpha $.
	Therefore $ \alpha\hat{f}_i\in \hat{\mathcal{F}}_{i} $ and $ \hat{\mathcal{F}}_{i} $ is a vector space.
\end{proof}
\begin{corollary}
	Consider $ \hat{\mathcal{F}}_{T} $, defined on the output sets $ \{\yset_i\}_{i\in\tset} $ that are vector spaces. Then, $ \hat{\mathcal{F}}_{T} $ is itself a vector space.
	\hfill$ \square $
\end{corollary}
\begin{lemma}
	Assume $ \hat{\mathcal{F}}_{T} $ is a vector space. Evaluation on a network $ (\rvert_{\mathcal{G}}) $ is a linear operator.
	\hfill$ \square $
\end{lemma}
\begin{proof}
	Consider oracle functions \mbox{$ \hat{f},\hat{g}\in\hat{\mathcal{F}}_{T} $}. Since $ \hat{\mathcal{F}}_{T} $ is a vector space, there is a $ \hat{h}\in\hat{\mathcal{F}}_{T} $ such that $ \hat{h} = \hat{f} + \hat{g} $. Define $ f = \hat{f}\rvert_{\mathcal{G}} $ and $ g = \hat{g}\rvert_{\mathcal{G}} $. Then $ h = \hat{h}\rvert_{\mathcal{G}} $ is such that
	\begin{align*}
	h_c(\mathbf{x}) &= \hat{h}_{i}\left(x_c; \mathbf{m}_c, \mathbf{x} \right)\\
	&= \hat{f}_{i}\left(x_c; \mathbf{m}_c, \mathbf{x} \right) + \hat{g}_{i}\left(x_c; \mathbf{m}_c, \mathbf{x} \right)\\
	&= f_c(\mathbf{x}) + g_c(\mathbf{x})
	\end{align*}
	for all $ c\in\mathcal{C} $ with $ i = \type{}(c) $. That is, $ (\hat{f} + \hat{g})\rvert_{\mathcal{G}} = \hat{f}\rvert_{\mathcal{G}} + \hat{g}\rvert_{\mathcal{G}} $.\\
	Similarly, for any $ \hat{f}\in\hat{\mathcal{F}}_{T} $ and any scalar $ \alpha $, there is a $ \hat{h}\in\hat{\mathcal{F}}_{T} $ such that $ \hat{h} = \alpha \hat{f} $. Define $ f = \hat{f}\rvert_{\mathcal{G}} $. Then $ h = \hat{h}\rvert_{\mathcal{G}} $ is such that
	\begin{align*}
	h_c(\mathbf{x}) &= \hat{h}_{i}\left(x_c; \mathbf{m}_c, \mathbf{x} \right)\\
	&= \alpha\hat{f}_{i}\left(x_c; \mathbf{m}_c, \mathbf{x} \right)\\
	&= \alpha f_c(\mathbf{x})
	\end{align*}
	for all $ c\in\mathcal{C} $ with $ i = \type{}(c) $. That is, $ (\alpha \hat{f})\rvert_{\mathcal{G}} = \alpha\cdot \hat{f}\rvert_{\mathcal{G}} $.\\
	Therefore, evaluation on a network $ (\rvert_{\mathcal{G}}) $ is a linear operator on $ \hat{\mathcal{F}}_{T} $.
\end{proof}
\begin{corollary}
	Assume $ \hat{\mathcal{F}}_{T} $ is a vector space. Then, \mbox{$ \mathcal{F}_{\mathcal{G}} = \hat{\mathcal{F}}_{T}\rvert_{\mathcal{G}} $} is also a vector space.
	\hfill$ \square $
\end{corollary}
\begin{corollary}
	Assume $ \hat{\mathcal{F}}_{T} $ is a vector space. Then, evaluating at a network $ \mathcal{G} $ (operator $ \rvert_{\mathcal{G}} $) partitions the space of functions $ \hat{\mathcal{F}}_{T} $ into affine planes parallel to the kernel (or nullspace) $ \ker(\rvert_{\mathcal{G}}) $ such that each plane represents the set of oracle functions that behave the same in that network, that is,
	\begin{align*}
	\hat{f}\rvert_{\mathcal{G}} = \hat{g}\rvert_{\mathcal{G}}
	\Longleftrightarrow
	\hat{f}-\hat{g}\in\ker(\rvert_{\mathcal{G}})
	\end{align*}
	for every $ \hat{f},\hat{g}\in\hat{\mathcal{F}}_{T} $.
	\hfill$ \square $
\end{corollary}
We now present synchrony properties for output vector spaces.
\begin{lemma}\label{lemma:invariant_vector_space}
	Assume $ \mathcal{F}_{\mathcal{G}} $ is a vector space. Then, for any partition \mbox{$ \mathcal{A}\leq\type{G} $}, the subset of $ \mathcal{F}_{\mathcal{G}} $ that is $ \mathcal{A} $-invariant is also a vector space.
	\hfill$ \square $
\end{lemma}
\begin{proof}
	Consider functions $ f,g\in \mathcal{F}_{\mathcal{G}}$ that are \mbox{$ \mathcal{A} $-invariant}. Defining, \mbox{$ h = f + g $}, if $ \mathbf{x} = P \overline{\mathbf{x}} $, where $ P $ is a partition matrix that represents $ \mathcal{A} $, we have
	\begin{align*}
	h(P\overline{\mathbf{x}}) &= f(P\overline{\mathbf{x}}) + g(P\overline{\mathbf{x}})\\
	&= P\overline{f}(\overline{\mathbf{x}}) + P\overline{g}(\overline{\mathbf{x}})\\
	&=P\left(\overline{f}(\overline{\mathbf{x}}) + \overline{g}(\overline{\mathbf{x}})\right)\\
	&=P\left(\overline{h}(\overline{\mathbf{x}})\right)
	\end{align*}
	so $ h $ is also $ \mathcal{A} $-invariant. Moreover,
	\begin{equation*}
	\alpha f(P\overline{\mathbf{x}}) = P\left(\alpha\overline{f}(\overline{\mathbf{x}})\right)
	\end{equation*}
	for any scalar $ \alpha $.
\end{proof}
\begin{corollary}
	Assume $ \mathcal{F}_{\mathcal{G}} = \hat{\mathcal{F}}_{T}\rvert_{\mathcal{G}}$ is a vector space. Then, for any partition \mbox{$ \mathcal{A}\leq\type{G} $}, the subset of $ \hat{\mathcal{F}}_{T} $ that is \mbox{$ \mathcal{A} $-invariant} is also a vector space.
	\hfill$ \square $
\end{corollary}
In a practical application, the nominal admissible function $ f^{\ast} $ that we desire in theory might not be the one that is actually realized. This motivates the interest in having some sort of local robustness so functions $ f $ that are sufficiently close to $ \hat{f} $ show similar properties.
\begin{corollary}
	Assume $ \mathcal{F}_{\mathcal{G}} $ is a \textbf{normed} vector space. Given a $ f^{\ast}\in \mathcal{F}_{\mathcal{G}} $, if we require that for some $ \varepsilon>0$, all the functions $ f\in\mathcal{F}_{\mathcal{G}} $ in the ball $ \lVert f - f^{\ast} \rVert < \varepsilon $ are \mbox{$ \mathcal{A} $-invariant}, then the whole $ \mathcal{F}_{\mathcal{G}} $ has to be \mbox{$ \mathcal{A} $-invariant}.
	\hfill$ \square $
\end{corollary}
This means that asking for local robustness in terms of \mbox{$ \mathcal{A} $-invariance} is the same as requiring the whole $ \mathcal{F}_{\mathcal{G}} $ to be \mbox{$ \mathcal{A} $-invariant}, that is, global \mbox{$ \mathcal{A} $-invariance}.\\
From \cref{thm:all_implies_balanced}, we know this can only be achieved when $ \mathcal{A} $ is a balanced partition in $ \mathcal{G} $.\\
We now present special cases for particular monoids in which we only require certain subsets of $ \mathcal{F}_{\mathcal{G}} $ to be \mbox{$ \mathcal{A} $-invariant} in order to enforce $ \mathcal{A} $ to be balanced.\\
The usual proof for \cref{thm:all_implies_balanced} in the unweighted and scalar-weighted cases, uses functions that are linear in the weights. This approach, however, does not scale well to general weight sets. Note that the analogous in this framework is to consider functions that are \textbf{additive in the weights}, that is,
\begin{align*}
p(w_1 \| w_2) =  p(w_1) + p(w_2)
\end{align*} 
If there is an annihilator in $ \mathcal{M} $, then $ p(w) =  0_{\mathcal{M}} $ for all \mbox{$ w\in\mathcal{M} $}. That is, only the trivial case for such functions exists. We now present an extension of the linear in the weights argument to a particular type of weight monoids for which it works.
\begin{remark}
	We always assume non-trivial output vector spaces. That is, $ \mathbb{Y}_i \neq \{0\} $. Otherwise, synchrony would always be trivially guaranteed but it would be a non-interesting particular case.
	\hfill$ \square $
\end{remark}
\begin{remark}
	If the state sets $ \xset_i $ only have one element, then it is irrelevant to talk about synchronicity in the first place. For this reason, we are assuming that the state sets are non-singleton and we can choose $ x_a \neq x_b $ with $ x_a,x_b \in\mathbb{X}_i $.
	\hfill$ \square $
\end{remark}

\begin{theorem}
	Consider non-trivial output vector spaces $ \{\yset_i\}_{i\in\tset} $ and assume that the edges are in the monoid $ \mathcal{M} = \left\langle\mset\vert E\right\rangle $, with $ \mset = \mathbb{R}\times\wset $ and
	\begin{align*}
	E = \{\lambda_1 w \| \lambda_2 w = (\lambda_1 +\lambda_2)w, \quad \forall \lambda_1,\lambda_2\in \mathbb{R}, w\in \wset \}
	\end{align*}
	where $ \wset $ is not necessarily countable.\\
	Consider the set of oracle components $ \hat{f}_i \ $, $ i\in\tset $, that are only dependent on neighbors that are in a specific state $ \overline{x}_k$, of the form 
	\begin{align*}
	\hat{f}_i\left(x;\sum \lambda_w w,\overline{x}_k\right) = \lambda_e\mathbf{v}, \quad \mathbf{v}\in \yset_{i}, \mathbf{v}\neq 0_{\yset_{i}}
	\end{align*}
	for some $ e\in\wset $.\\
	If the subset of oracle functions in $ \mathcal{F}_{\mathcal{G}} $ that are built with $ \hat{f}_i $'s of the type above is \mbox{$ \bp $-invariant}, then $ \bp $ is balanced in $ \mathcal{G} $. 
	\hfill$ \square $
\end{theorem}
\begin{proof}
	Assume $ \bp $ is not balanced. Then, there are \mbox{$ \bp(c) = \bp(d) $} such that $ \mathbf{m}_c P \neq \mathbf{m}_d P $. Consider $ k $ one of the colors in which they differ. That is, $ \left[\mathbf{m}_c P\right]_k \neq \left[\mathbf{m}_d P\right]_k $.\\
	Moreover, choose some state $ \mathbf{x}\in\Delta_{\bp}^{\xset} $, that is, $ \mathbf{x} = P\overline{\mathbf{x}} $, such that $ \overline{x}_k $ is different from all other entries of $ \overline{\mathbf{x}} $.\\
	An element of the monoid $ \mathcal{M} $ can be written as linear combination over a finite  subset of elements in $ \wset $, that is, $ \sum \lambda_w w $. If $ \left[\mathbf{m}_c P\right]_k \neq \left[\mathbf{m}_d P\right]_k $, then they differ on the associated coefficient of at least one element $ e\in\wset $. Then, there is an $ \hat{f}_i $ as defined above, sensitive to that element $ e $, so that
	$ 
	\hat{f}_{i}\left(x; \left[\mathbf{m}_c P\right]_k, \overline{x}_k \right) 
	=
	\lambda_e^c\mathbf{v}
	\neq
	\lambda_e^d\mathbf{v}
	=
	\hat{f}_{i}\left(x; \left[\mathbf{m}_d P\right]_k, \overline{x}_k \right) 
	$.\\
	Then, we have an $ f\in\mathcal{F}_{\mathcal{G}} $ and $ \overline{\mathbf{x}} $ such that $ f\left(P\overline{\mathbf{x}}\right) \not\in \Delta_{\bp}^{\yset} $. Therefore, the only way to not be able to find such an $ f $, is for $ \bp $ to be balanced.
\end{proof}
The next result is valid for systems in which the weight set allows for the existence of an annihilator. However, the monoid is almost free, in the sense that its congruence relation does not define further equivalence classes.
\begin{theorem}
	Consider non-trivial output vector spaces $ \{\yset_i\}_{i\in\tset} $ and assume that the edges are either on a free monoid $ \mathcal{M} = \left\langle\wset\vert\right\rangle $ or the result of adding an annihilator to a free monoid. That is, $ \mathcal{M} = \left\langle \{a\} \cup \wset\vert E\right\rangle $, with
	\begin{align*}
	E = \{w\|a = a, \quad \forall w\in \mathcal{M}\}
	\end{align*}
	where $ \wset $ is not necessarily countable.\\
	Consider the set of oracle components $ \hat{f}_i \ $, $ i\in\tset $, that are only dependent on neighbors that are in a specific state $ \overline{x}_k$, of the form 
	\begin{align*}
	\hat{f}_i\left(x;\sum w,\overline{x}_k\right) = \mathbf{v}\prod p(w), \quad \mathbf{v}\in \yset_{i}, \mathbf{v}\neq 0_{\yset_{i}}
	\end{align*}
	If the subset of oracle functions in $ \mathcal{F}_{\mathcal{G}} $ that are built with $ \hat{f}_i $'s of the type above is \mbox{$ \bp $-invariant}, then $ \bp $ is balanced in $ \mathcal{G} $. 
	\hfill$ \square $
\end{theorem}
\begin{proof}
	Assume $ \bp $ is not balanced. Then, there are \mbox{$ \bp(c) = \bp(d) $} such that $ \mathbf{m}_c P \neq \mathbf{m}_d P $. Consider $ k $ one of the colors in which they differ. That is, $ \left[\mathbf{m}_c P\right]_k \neq \left[\mathbf{m}_d P\right]_k $.\\
	Moreover, choose some state $ \mathbf{x}\in\Delta_{\bp}^{\xset} $, that is, $ \mathbf{x} = P\overline{\mathbf{x}} $, such that $ \overline{x}_k $ is different from all other entries of $ \overline{\mathbf{x}} $.\\
	An element of the monoid $ \mathcal{M} $ can be written as a finite sum over elements in $ \wset $. Call the support, that is, the elements that appear at least once in $ \left[\mathbf{m}_c P\right]_k $ as $ w_1,\ldots,w_n $ and the support of $ \left[\mathbf{m}_d P\right]_k $ as $ v_1 ,\ldots, v_m $. We can, with some function $ p $, assign to each distinct element of the union of both sets, a distinct prime number, with the exception of the zero element $ 0_\mset $ and a possible annihilator $ a $, in which we have instead that $ p(0_\mathcal{M}) = 1 $ and $ p(a) = 0 $.
	If $ \left[\mathbf{m}_c P\right]_k \neq \left[\mathbf{m}_d P\right]_k $, then, there is an $ \hat{f}_i $ as defined above, so that $ 
	\hat{f}_{i}\left(x; \left[\mathbf{m}_c P\right]_k, \overline{x}_k \right) 
	\neq
	\hat{f}_{i}\left(x; \left[\mathbf{m}_d P\right]_k, \overline{x}_k \right) 
	$.\\
	Then, we have an $ f\in\mathcal{F}_{\mathcal{G}} $ and $ \overline{\mathbf{x}} $ such that $ f\left(P\overline{\mathbf{x}}\right) \not\in \Delta_{\bp}^{\yset} $. Therefore, the only way to not be able to find such an $ f $, is for $ \bp $ to be balanced.
\end{proof}
\begin{remark}
	Note that additional conditions on $ E $, that defines congruence relation of the monoid, might invalidate this approach, e.g., $ w_1 \| w_2 = w_3 \| w_4 $, where all weights are different.
	\hfill$ \square $
\end{remark}

\section{Quotients}\label{sec:quotients}
In this section we describe how the behavior of a network $ \mathcal{G} $ when evaluated at some polydiagonal $ \Delta_\bp$ for some balanced partition $ \bp $ can be described by a smaller network $ \mathcal{Q} $.
\begin{definition}
	Consider a network $ \mathcal{G} $ defined on a cell set $ \mathcal{C}_{\mathcal{G}} $ with a cell type partition $ \type{G} $ and an in-adjacency matrix $ M $. Take a partition $ \bp $ balanced on $ \mathcal{G} $.\\
	The \textbf{quotient network} $ \mathcal{Q} $ of $ \mathcal{G} $ over $ \bp $, denoted \mbox{$ \mathcal{Q} = \mathcal{G}/\bp $} is defined on a cell set \mbox{$ \mathcal{C}_{\mathcal{Q}} = \mathcal{C}_{\mathcal{G}}/\bp $} with a cell type partition \mbox{$ \type{Q} = \type{G}/\bp $} and an in-adjacency matrix $ Q $ given by \mbox{$ M P = P Q $}, where $ P $ represents $ \bp $.
	\hfill$ \square $
\end{definition}
\begin{definition}
	Consider networks $ \mathcal{G} $ and $\mathcal{Q} $ such that \mbox{$ \mathcal{Q} = \mathcal{G}/\bp $} and some \mbox{$ \mathcal{G} $-admissible} function $ f = \hat{f}\rvert_{\mathcal{G}} $, with $ \hat{f}\in\hat{\mathcal{F}}_{T} $.\\
	The \textbf{quotient function} $ g $ of $ f $ over the balanced partition $ \bp $, denoted $ g = f/\bp $ is given by $ g = \hat{f}\rvert_{\mathcal{Q}}  $.
	\hfill$ \square $
\end{definition}
\begin{lemma}
	The quotient function $ g = f/\bp $ is well-defined since is does not depend on the particular choice of oracle function. That is, for any $ \hat{f}^1,\hat{f}^2 \in\hat{\mathcal{F}}_{T} $
	\begin{align*}
	\hat{f}^1\rvert_{\mathcal{G}} = \hat{f}^2\rvert_{\mathcal{G}} \implies
	\hat{f}^1\rvert_{\mathcal{Q}} = \hat{f}^2\rvert_{\mathcal{Q}} 
	\end{align*}
	that is, a particular $ f\in\mathcal{F}_{\mathcal{G}} $ implies a unique $ g\in\mathcal{F}_{\mathcal{Q}} $. 
	\hfill$ \square $
\end{lemma}
\begin{proof}
	By assumption of $ \hat{f}^1\rvert_{\mathcal{G}} = \hat{f}^2\rvert_{\mathcal{G}} $ we have,
	\begin{align*}
	\hat{f}^1_{i}\left(\overline{x}_k; \mathbf{m}_c, P\overline{\mathbf{x}} \right)
	=
	\hat{f}^2_{i}\left(\overline{x}_k; \mathbf{m}_c, P\overline{\mathbf{x}} \right)
	\end{align*}
	for all $ c \in \mathcal{C}_{\mathcal{G}}$ with $ k = \bp(c) $.\\
	Since $ \mathbf{q}_k = \mathbf{m}_c P $, where the weight vector $ \mathbf{q}_k $ is the $ k^{th} $ row of $ Q $, an \mbox{in-adjacency} matrix of $ \mathcal{Q} $. This implies
	\begin{align*}
	\hat{f}^1_i(\overline{x}_k;\mathbf{q}_k,\overline{\mathbf{x}})
	=
	\hat{f}^2_i(\overline{x}_k;\mathbf{q}_k,\overline{\mathbf{x}})
	\end{align*}
	for any $ k\in\mathcal{C}_{\mathcal{Q}} $, with \mbox{$ i = \type{Q}(k) $}. That is $ \hat{f}^1\rvert_{\mathcal{Q}} = \hat{f}^2\rvert_{\mathcal{Q}} $.
\end{proof}
In \cref{sec:invariant_synchrony} it was shown that any admissible $ f\in\mathcal{F}_{\mathcal{G}} $ when evaluated on $ \Delta_\bp$ can be determined by a simpler function, related to it by \cref{eq:f_in_invariant_space}. The following results show that this is what connects the quotient network and quotient function.
\begin{theorem}
	Consider a network $ \mathcal{G} $ and a partition $ \bp $ balanced on it. Let $ f\in\mathcal{F}_{\mathcal{G}} $.
	The function obtained by constraining $ f $ to the polydiagonal $ \Delta_\bp$, is the quotient function $ g = f/\bp $, that is,
	\begin{align*}
	f(P\overline{\mathbf{x}}) = Pg(\overline{\mathbf{x}})
	\end{align*}
	\hfill$ \square $
\end{theorem}
\begin{proof}
	Consider some oracle function $ \hat{f}\in\hat{\mathcal{F}}_{T} $ such that $ f = \hat{f}\rvert_{\mathcal{G}} $ and $ g = \hat{f}\rvert_{\mathcal{Q}} $. Then, when \mbox{$ \mathbf{x}\in\Delta_\bp $}, that is $ \mathbf{x} = P\overline{\mathbf{x}} $, we have that
	\begin{align*}
	f_c(P\overline{\mathbf{x}}) &= \hat{f}_{i}\left(\overline{x}_k; \mathbf{m}_c, P\overline{\mathbf{x}} \right)\\
	& =\hat{f}_i(\overline{x}_k;\mathbf{q}_k,\overline{\mathbf{x}})\\
	& = g_k(\overline{\mathbf{x}})
	\end{align*}
	since $ \mathbf{q}_k = \mathbf{m}_c P $, for all $ c \in \mathcal{C}_{\mathcal{G}}$ with \mbox{$ i = \type{G}(c) $} and \mbox{$ k = \bp(c) $}. Therefore, the function \mbox{$ g =  (g_k)_{k\in\mathcal{C}_{\mathcal{Q}}}$} is related to $ f $ by equation \cref{eq:f_in_invariant_space}.
\end{proof}
\begin{exmp}
	Consider the given partition $ \mathcal{A} = \{ \{1,2\},\{3\} \} $ on the \textbf{CCN} of \cref{exmp:admiss_funcs} (\cref{fig:Gadmissibility}). One partition matrix of $ \mathcal{A} $ is 
	\begin{equation}P =
	\begin{bmatrix}
	1 & 0\\
	1 & 0\\
	0 & 1
	\end{bmatrix}
	\end{equation}
	in which each column identifies one of the colors of the partition.
	From this we obtain the product
	\begin{equation}MP =
	\begin{bmatrix}
	1 & 1\\
	1 & 1\\
	2 & 1
	\end{bmatrix}
	\label{eq:MP_example}
	\end{equation}
	Note that rows $ 1 $ and $ 2 $ are the same. That means that for any admissible $ f $ we have $ f_1(\mathbf{x}) = f_2(\mathbf{x}) $ when $ x_1 = x_2 $.\\
	Observe that this is in agreement with the functional form we wrote in \cref{eq:func_match1,eq:func_match2}.\\
	Since the rows of $ MP $ respect an equality relationship according to $ \mathcal{A} $, then $ \mathcal{A} $ is balanced and there is a quotient matrix $ Q $ that obeys the balanced condition \cref{eq:balanced_cond_matrix_form}. In fact, the quotient matrix $ Q  $ is
	\begin{equation}
	Q =
	\begin{bmatrix}
	1 & 1\\
	2 & 1
	\end{bmatrix}
	\end{equation}
	which is directly obtained from $ MP $ by compressing its rows according to $ \mathcal{A} $.\\
	The behavior of this \textbf{CCN} when $ x_1 = x_2 $ is then described by the smaller \textbf{CCN} given by the quotient matrix $ Q $ which is represented in \cref{fig:Gadmissibility_quotientb}.
	\begin{figure}[h]
		\centering
		\begin{subfigure}[t]{0.23\textwidth}
			\centering
			\input{Gadmissibility_example_color_part.tikz}
			\caption{Original}
			\label{fig:Gadmissibility_quotienta}
		\end{subfigure}
		\begin{subfigure}[t]{0.23\textwidth}
			\centering
			\input{Gadmissibility_example_quotient.tikz}
			\caption{Quotient}
			\label{fig:Gadmissibility_quotientb}
		\end{subfigure}
		\caption{Color-coded network of \cref{fig:Gadmissibility} and its quotient over the balanced partition $ \{ \{1,2\},\{3\} \} $}
		\label{fig:Gadmissibility_quotient}
	\end{figure}
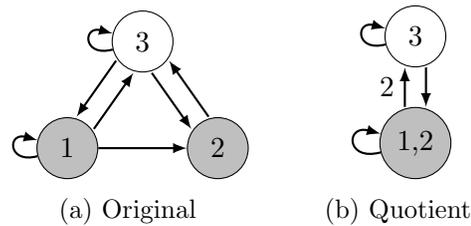
	The coloring is a way of representing the partition $ \mathcal{A} = \{ \{1,2\},\{3\} \} $ over which the quotient is done. Note that in both \cref{fig:Gadmissibility_quotienta,fig:Gadmissibility_quotientb} each gray cell receives one connection from a gray cell and one connection from a white cell. On the other hand, each white cell receives a connection from a white cell and two connections from a gray cell.
	The function $ g = f/\bp $ has the following structure
	\begin{align}
	g_{12}(\mathbf{x}) &= \hat{f}(x_{12};\begin{bmatrix} 1 & 1\end{bmatrix},\mathbf{x})\\
	g_3(\mathbf{x}) &= \hat{f}(x_3;\begin{bmatrix} 2 & 1\end{bmatrix},\mathbf{x})
	\end{align}
	where $ \hat{f}\in\hat{\mathcal{F}}_{T}  $ is any oracle function such that $ f = \hat{f}\rvert_{\mathcal{G}} $.
	\hfill$ \square $
\end{exmp}
\begin{remark}
	Note that finding a balanced partition from its graph representation or its matrix $ M $ is not obvious. See \cref{exmp:no_balanced}.
	\hfill$ \square $
\end{remark}
\begin{exmp}\label{exmp:no_balanced}
	Consider the following network illustrated in \cref{fig:CCN_chain}.
	\begin{figure}[h]
		\centering
		\input{CCN_chain.tikz}
		\caption{Chain \textbf{CCN}}
		\label{fig:CCN_chain}
	\end{figure}
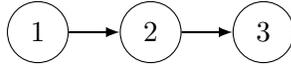
	Since cells $ 2 $ and $ 3 $ have the same type of input it might be tempting to think that $ \mathcal{A} = \{ \{1\},\{2,3\} \} $ should be balanced. Note, however, that the rows of the corresponding matrix $ MP $ \cref{eq:MP} do not respect the row equalities according to $ \mathcal{A} $, which means that it is not balanced.
	\begin{equation}MP=
	\begin{bmatrix}
	0 & 0\\
	1 & 0\\
	0 & 1
	\end{bmatrix}
	\label{eq:MP}
	\end{equation}
	Another way to see this is to color the cells according to the partition (\cref{fig:CCN_chain_unbalanced}) and see that cells with the same color do not have equivalent colored input sets. 
	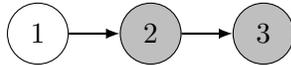
\begin{figure}[h]
		\centering
		\input{CCN_chain_unbalanced.tikz}
		\caption{Unbalanced coloring}
		\label{fig:CCN_chain_unbalanced}
	\end{figure}
	Note that cells $ 2 $ and $ 3 $ are both gray but one of them receives one edge from a white cell and the other receives one edge from a gray cell. Therefore, this coloring (partition) is not balanced. In fact, it can be easily seen that the only balanced partition of this network is the trivial one.
	\hfill$ \square $
\end{exmp}
The following result shows that the quotient operation is transitive.
\begin{lemma}\label{lemma:quotient_lemma_transitive}
	Consider networks $ \mathcal{G}_0 $, $ \mathcal{G}_1 $, $ \mathcal{G}_2 $ such that there are balanced partitions $ \bp_{01} $, $ \bp_{12} $ such that $ \mathcal{G}_1 = \mathcal{G}_0/{\bp_{01}} $ and $ \mathcal{G}_2 = \mathcal{G}_1/{\bp_{12}} $. Then, there is a partition $ \bp_{02} $ with \mbox{$ \bp_{01} \leq \bp_{02} $} such that \mbox{$ \mathcal{G}_2 = \mathcal{G}_0/{\bp_{02}} $}. Furthermore, \mbox{$ P_{02} = P_{01}P_{12} $} where the matrices represent the corresponding indexed partitions.
	\hfill$ \square $
\end{lemma}
\begin{proof}
	Name the cell type partitions of the networks $ {\type{}}_0 $, $ {\type{}}_1 $, $ {\type{}}_2 $ accordingly. Then,
	\begin{align*}
	{\type{}}_0 = P_{01}{\type{}}_1 ,\quad {\type{}}_1 = P_{12}{\type{}}_2
	\end{align*}
	implies
	\begin{align*}
	{\type{}}_0 = (P_{01}P_{12}){\type{}}_2
	\end{align*}
	Let $ M_0 $, $ M_1 $, $ M_2 $ represent the in-adjacency matrices of the networks. From being balanced we know that
	\begin{align*}
	M_0 P_{01}  = P_{01} M_1, \quad	M_1 P_{12} = P_{12} M_2 
	\end{align*}
	by multiplying the first equality by $ P_{12} $ on the right
	\begin{align*}
	M_0 (P_{01}P_{12})  &= P_{01}(M_1 P_{12}) \\
	&= (P_{01} P_{12} )M_2 
	\end{align*}
	which means that
	\begin{align*}
	\mathcal{G}_2 = \left(\mathcal{G}_0/{\bp_{01}}\right)/{\bp_{12}} = \mathcal{G}_0/{\bp_{02}}
	\end{align*}
	and that \mbox{$ P_{02} = P_{01}P_{12} $}.
\end{proof}
The next results shows how two different quotients of the same network can be related by one of them being a quotient of the other
\begin{lemma}\label{lemma:quotient_lemma_leq}
	Consider networks $ \mathcal{G}_0 $, $ \mathcal{G}_1 $, $ \mathcal{G}_2 $ such that there are balanced partitions $ \bp_{01} $, $ \bp_{02} $ such that $ \mathcal{G}_1 = \mathcal{G}_0/{\bp_{01}} $ and $ \mathcal{G}_2 = \mathcal{G}_0/{\bp_{02}} $. If $ \bp_{01} \leq \bp_{02} $ there is a balanced partition $ \bp_{12} $ such that \mbox{$ \mathcal{G}_2 = \mathcal{G}_1/{\bp_{12}} $}. Furthermore, \mbox{$ P_{02} = P_{01}P_{12} $} where the matrices $ P $ represent the correspondingly indexed partitions.
	\hfill$ \square $
\end{lemma}
\begin{proof}
	Name the cell type partitions of the networks $ {\type{}}_0 $, $ {\type{}}_1 $, $ {\type{}}_2 $ accordingly. Then,
	\begin{align*}
	{\type{}}_0 = P_{01}{\type{}}_1 ,\quad {\type{}}_0 = P_{02}{\type{}}_2
	\end{align*}
	from $  \bp_{01} \leq \bp_{02} $ we know that there is a partition matrix $ P_{12} $ such that 
	\begin{align*}
	P_{02} = P_{01}P_{12}
	\end{align*}
	replacing $ P_{02} $ in the second equality
	\begin{align*}
	{\type{}}_0 &=  P_{01}P_{12}{\type{}}_2
	\end{align*}
	using the first equality
	\begin{align*}
	P_{01}{\type{}}_1 &= P_{01}(P_{12}{\type{}}_2)
	\end{align*}
	Since $ P_{01} $ has full column rank, it can be canceled on the left
	\begin{align*}
	{\type{}}_1 &= P_{12}{\type{}}_2
	\end{align*}
	Let $ M_0 $, $ M_1 $, $ M_2 $ represent the in-adjacency matrix of the networks. From being balanced we know that
	\begin{align*}
	M_0 P_{01}  =  P_{01} M_1, \quad	M_0 P_{02}  =  P_{02} M_2 
	\end{align*}
	replacing $ P_{02} $ in the second equality
	\begin{align*}
	(M_0 P_{01})P_{12} &= P_{01} P_{12} M_2 
	\end{align*}
	which by the first balanced equality
	\begin{align*}
	(P_{01} M_1)P_{12} &= P_{01} P_{12} M_2  
	\end{align*}
	now, canceling $ P_{01} $ on the left is allowed
	\begin{align*}
	M_1 P_{12}   &= P_{12} M_2 
	\end{align*}
	which means that $ \bp_{12} $, represented by $ P_{12} $ is balanced on $ \mathcal{G}_1 $ and
	\begin{align*}
	\mathcal{G}_2 = \mathcal{G}_1/{\bp_{12}}
	\end{align*}
\end{proof}

\section{Lattice of balanced partitions}\label{sec:lattice}
This sections presents the properties of $ \Lambda_{\mathcal{G}} $, which denotes the set of all balanced partitions of a given network $ \mathcal{G} $.\\
In~\cite{stewart2007lattice} it is shown that $ \Lambda_{\mathcal{G}} $ forms a lattice under the sub-partition operation $ \leq $ as described in \cref{eq:refinement_def}. That is, it forms a partially ordered set such that for any $ \bp_1, \bp_2\in\Lambda_{\mathcal{G}} $, there also exist in $ \Lambda_{\mathcal{G}} $ partitions $ \bp_1 \vee \bp_2 $ and $ \bp_1 \wedge \bp_2 $ which are the \textbf{least upper bound} or \textbf{join} and the \textbf{greatest lower bound} or \textbf{meet}.
This means that there is a \textbf{maximal} ($ \top $) and a \textbf{minimal} ($ \bot $) balanced partitions, the last of which we already know to be the trivial partition.\\
The coarsest invariant refinement (\textbf{CIR}) algorithm, is a polynomial-time algorithm that was first developed in~\cite{aldis2008polynomial} with the goal of finding the maximal balanced partition. Recently, in~\cite{neuberger2019invariant} it was noted that this algorithm does more that just finding the maximal balanced partition. In fact, given any input partition, it outputs the greatest balanced partition that is finer ($ \leq $) than the input one. Therefore, the maximal partition is given by $ \top = cir(\type{})$. 
We now show that the meet ($ \wedge $) and join ($ \vee $) are also defined.
\begin{lemma}
	For every pair $ \bp_1, \bp_2 \in\Lambda_{\mathcal{G}} $ there is a least upper bound or join $ \bp_3 \in\Lambda_{\mathcal{G}} $ denoted by \mbox{$ \bp_3 = \bp_1 \vee \bp_2$}.
	\hfill $ \square $
\end{lemma}
\begin{proof}
	Any partition $ \bp $ that is simultaneous coarser than $ \bp_1 $ and $ \bp_2 $ has to obey
	\begin{align*}
	\begin{cases}
	\bp_1(c) &= \bp_1(d)\\
	&\text{or}\\
	\bp_2(c) &= \bp_2(d)
	\end{cases}
	\implies
	\bp(c) = \bp(d)
	\end{align*}
	For such partition, any chain of cells \mbox{$ c=c_1,\ldots,c_k = d $} such that either $ \bp_1(c_{i}) = \bp_1(c_{i+1})$ or $ \bp_2(c_{i}) = \bp_2(c_{i+1}) $ implies that $ \bp(c) = \bp(d) $. The finest such partition $ \bp_3 $ is the one such that $ \bp_3(c) = \bp_3(d) $ if and only if there is such a chain.\\
	We now show that $ \bp_3 $ is balanced. Choose any two colors $ A, B\in \bp_3 $. Since $ \bp_1, \bp_2 $ are both sub-partitions of $ \bp_3 $ there are colors $ b_1^{1},\ldots,b_{k_1}^{1} \in \bp_1$ and $ b_1^{2},\ldots,b_{k_2}^{2} \in \bp_2$ such that $ B =  \bigcup_{i=1}^{k_1} b_{i}^{1} = \bigcup_{i=1}^{k_2} b_{i}^{2}$. For any $ c,d\in A $ such that $ \bp_i(c) = \bp_i(d) $ for some $ i\in\{1,2\} $ we have that
	\begin{align*}
	\sum_{e\in b_{j}^{i}} w_{ce}
	=
	\sum_{e\in b_{j}^{i}} w_{de}
	\quad
	\forall j\in \{1,\ldots,k_i\}
	\end{align*}
	which implies
	\begin{align*}
	\sum_{e\in B} w_{ce}
	=
	\sum_{e\in B} w_{de}
	\end{align*}
	since for each link of the chain this value is preserved, it is also preserved across the whole chain and therefore the whole color $ A $. This applies to every pair of colors $ A,B\in\bp $ which means that $ \bp $ is balanced.
\end{proof}
\begin{lemma}
	For every pair $ \bp_1, \bp_2 \in\Lambda_{\mathcal{G}} $ there is a greatest lower bound or meet \mbox{$ \bp_{3} \in\Lambda_{\mathcal{G}} $} denoted by \mbox{$ \bp_3 = \bp_1 \wedge \bp_2$}.
	\hfill $ \square $
\end{lemma}
\begin{proof}
	Any partition $ \bp $ that is simultaneous finer than $ \bp_1 $ and $ \bp_2 $ has to obey
	\begin{align*}
	\bp(c) = \bp(d)
	\implies
	\begin{cases}
	\bp_1(c) = \bp_1(d)\\
	\bp_2(c) = \bp_2(d)
	\end{cases}
	\end{align*}
	call $ \bp_{12} $ the coarsest such partition, created by making the implication into an equivalence. Since such a partition is unique, $ \bp_3 $ is given by $ \bp_3 = cir(\bp_{12} ) $. Note that there was no reason to believe that $ \bp_{12} $ was itself balanced and can be verified empirically not to be.
\end{proof}
The following result relates the lattice of a network with the lattice of one of its quotients.
\begin{lemma}\label{lemma:super_partition_thm}
	Consider networks $\mathcal{G} $ and $ \mathcal{Q} $ related by \mbox{$ \mathcal{Q} = \mathcal{G}/\bp $}. Then, there is a one-to-one correspondence between the elements of $ \Lambda_{\mathcal{Q}} $ and the elements of $ \Lambda_{\mathcal{G}} $ that are coarser than $ \bp $. This relation is given as $ P_{\mathcal{G}} = P_{\bp} P_{\mathcal{Q}}  $ where $ P_{\mathcal{G}} $, $ P_{\mathcal{Q}} $ represent partitions in $ \Lambda_{\mathcal{G}}$, $ \Lambda_{\mathcal{Q}}$ respectively. Therefore, we say that $ \Lambda_{\mathcal{Q}} = \Lambda_{\mathcal{G}}/\bp $.
	\hfill $ \square $
\end{lemma}
\begin{proof}
	\Cref{lemma:quotient_lemma_transitive} shows that for any partition \mbox{$ \bp_{12}\in\Lambda_{\mathcal{Q}} $} there is a partition \mbox{$ \bp_{02}\in\Lambda_{\mathcal{G}} $}, coarser than $ \bp $, such that \mbox{$ P_{02} =  P_{\bp} P_{12}  $}. Conversely, \cref{lemma:quotient_lemma_leq} shows that for any \mbox{$ \bp_{02}\in\Lambda_{\mathcal{G}} $} that is coarser than $ \bp $ there is a partition \mbox{$ \bp_{12}\in\Lambda_{\mathcal{Q}} $} such that they are related to each other in the same way as before. 
\end{proof}
The problem of trying to get an exhaustive list of the elements of a lattice $ \Lambda $ can be potentially intractable. Note that given a partition, it is easy and efficient to verify whenever it is balanced. However, even for relatively small networks, the set of all possible partitions is simply too large to do an exhaustive search on it. The number of partitions on a set is given by the Bell numbers $ B_n $ (also called exponential numbers), referenced in the online database OEIS (The On-Line Encyclopedia of Integer Sequences by the code A000110~\cite{OEISA000110}. A method to reduce the search space is described in~\cite{neuberger2019invariant}. The algorithm, called `\textbf{SPLIT} and \textbf{CIR}', uses the idea that instead of testing all $\left( \prod_{i=1}^{r} (B_{k_i} -1)\right) $ sub-partitions of $ \top $, one can apply the \textbf{CIR} method to its $\left( \sum_{i=1}^{r} 2^{k_i-1} \right)$ immediate descendants and then repeat, finding all the balanced partitions layer by layer.\\
The specific gains of this approach are difficult to analyze and can be highly dependent of the particular network of study (e.g., the number of layers in the lattice $ \Lambda_{\mathcal{G}} $). A worst-case scenario (all $ B_n $ partitions balanced) evaluation could be too pessimistic and a bad metric to decide if it would be an approach of interest for application in a real-world network.\\
This worst-case is scenario is a shortcoming that is common to all approaches that try to find all the balanced partitions in an explicit exhaustive way. Consider for example an all-to-all uniform-connection network of size $ n $, with a single cell type. In this case, even for a relatively small $ n $, the lattice would be too large to enumerate its ($ B_n $) elements or draw any schematic, yet it can be described in one simple sentence (i.e., ``every partition is balanced"). Ironically, the simplest network, whose lattice is the easiest to determine corresponds exactly to the worst-case scenario of such methods.

\section{CIR algorithm improvement}\label{sec:CIR_algorithm_improvement}
In this section we describe our improvement of the \textbf{CIR} algorithm that works with general weight sets and has a worst-case complexity of $ \mathbf{O}(\vert\mathcal{C}\vert^3) $ in the case of a dense graph and $ \mathbf{O}(\vert\mathcal{C}\vert^2) $ in the sparse case.\\
Consider a network represented by a matrix $ M $ together with an initial partition $ \mathcal{A}_0\leq\type{} $ represented by matrix $ P_0 $, of which we want to find the coarsest refinement (e.g., make $ P_0 $ the characteristic matrix of $ \type{}$ if the goal is to find the maximal balanced partition $ \top $).
\subsection{Method}
The idea of this algorithm is to start with the initial partition $ \mathcal{A}_0 $ and progressively refine it in a conservative manner.
That is, given a partition $ \mathcal{A}_i $, we construct a partition \mbox{$ \mathcal{A}_{i+1} \leq \mathcal{A}_i $} such that any balanced partition finer than $ \mathcal{A}_i $ is also finer than $ \mathcal{A}_{i+1} $. We create $ \mathcal{A}_{i+1} $ by taking each color of $ \mathcal{A}_i $ and splitting its cells according to whenever their corresponding rows in $  M P_i  $ match or not.
If \mbox{$ \mathcal{A}_{i+1} = \mathcal{A}_i $} the algorithm has converged and we found $ \mathcal{A}_{i} = cir(\mathcal{A}_0) $, otherwise we continue iterating.
\begin{lemma}
	According to the described iterative method, any balanced partition finer than $ \mathcal{A}_i $ is also finer than $ \mathcal{A}_{i+1} $.
	\hfill $ \square $
\end{lemma}
\begin{proof}
	Assume that there are cells $ c $, $ d $ such that \mbox{$ \mathcal{A}_i(c) = \mathcal{A}_i(d) $} but rows $ c $ and $ d $ of $ M P_i  $ do not match perfectly (assume on $ k^{th} $ column). Note that the $ k^{th} $ color of $ \mathcal{A}_i $ will correspond to either a color, or a union of colors of any balanced partition finer than $ \mathcal{A}_i $. This means that no matter what refinement happens, the cells $ c $ and $ d $ will have no chance of having the same color in a balanced refinement, since if the sum of the parts is different, it will not be possible for the parts themselves to match. Therefore, any balanced partition finer than $ \mathcal{A}_i $ is also finer than $ \mathcal{A}_{i+1} $.	
\end{proof}
\begin{remark}
	Note that if at a certain iteration no more refinement happens, that means that the balanced condition \cref{eq:balanced_cond} has been achieved and we found $ cir(\mathcal{A}_0) $.
	\hfill $ \square $
\end{remark}
\begin{lemma}
	The iterative procedure always converges in at most \mbox{$ \vert\mathcal{C}\vert - rank(\mathcal{A}_0) $} iterations.
	\hfill $ \square $
\end{lemma}
\begin{proof}
	Note that in each iteration, either the rank of the partition increases or the algorithm stops because a balanced partition was achieved. In the worst case scenario, the rank increases by one until the trivial partition is reached. Therefore, the algorithm always converges in at most \mbox{$ \vert\mathcal{C}\vert - rank(\mathcal{A}_0) $} iterations.
\end{proof}
Since this algorithm always converges, this shows by construction that $ cir(\mathcal{A}_0) $ exists. That is, for any partition $ \mathcal{A}_0 $, there is a unique balanced partition $ \mathcal{A}_i = cir(\mathcal{A}_0) $ such $ \mathcal{A}_i \leq \mathcal{A}_0 $ and $ \bp \leq \mathcal{A}_i $ for any balanced partition $ \bp $ such that $ \bp \leq \mathcal{A}_0 $.
\subsection{Efficient implementation and cost analysis}
Note that a partition matrix $ P $ on a set of cells $ \mathcal{C} $ can be efficiently represented by a vector of size $ \vert\mathcal{C}\vert $ as seen in \cref{exmp:pariti_vec_func}.
Calculating the product $ M P_i $ consists on summing ($ \| $) certain elements of $ M $ according to the pattern described in $ P_i $.
To compare rows of $ M P_i $ previous works considered a quadratic cost which was the bottleneck of the algorithm. If the appropriate data structure (hash table) is used, such operation is of the order $ \mathbf{O}(rank(\mathcal{A}_i) )$.
A pseudo-code description of the algorithm implementation is presented in \cref{alg:cir_alg,alg:cir_iter}.
\begin{algorithm}
	\begin{algorithmic}
		\STATE $ M \gets \text{CCN matrix}$
		\STATE $ p_0 \gets \text{initial partition vector} $
		\STATE $ r_0 \gets \text{rank of } p_0 $
		\STATE $ p_{new} \gets p_{0}$
		\STATE $ r_{new} \gets r_{0}$
		\REPEAT
		\STATE $ p_{old} \gets p_{new} $
		\STATE $ r_{old} \gets r_{new} $
		\STATE $ (p_{new}, r_{new}) \gets cir\_iteration(M,p_{old}, r_{old}) $
		\UNTIL $ r_{new} \mathrel{=}= r_{old}$
	\end{algorithmic}
	\caption{CIR algorithm}
	\label{alg:cir_alg}
\end{algorithm}
\begin{algorithm}
	\begin{algorithmic}
		\STATE $ M \gets \text{CCN matrix}$
		\STATE $ p_{old} \gets \text{previous partition vector} $
		\STATE $ r_{old} \gets \text{rank of } p_{old} $
		\STATE $ p_{new} \gets \text{new partition vector} $
		\STATE $ r_{new} \gets 0 $
		\FOR{$r = 1\colon\vert\mathcal{C}\vert$}
		\STATE {$ v \gets \text{zero vector of size } r_{old} $}
		\FOR{$ c \colon (r,c) \in \mathcal{E} $}
		\STATE{$ v[p_{old}(c)] \gets v[p_{old}(c)] + M(r,c) $}
		\ENDFOR
		\STATE $ s \gets vec2string([p_{old}(r), v] )$
		\STATE $ value \gets hash\_table.find(s) $ 
		\IF{value NOT\_FOUND}
		\STATE $ \mathrel{r_{new}} \gets \mathrel{r_{new}} + 1 $
		\STATE $ p_{new}[r] \gets r_{new} $
		\STATE $ hash\_table.insert(s,r_{new}) $
		\ELSE
		\STATE $ p_{new}[r] \gets value $
		\ENDIF
		\ENDFOR
	\end{algorithmic}
	\caption{CIR iteration}
	\label{alg:cir_iter}
\end{algorithm}
\begin{lemma}
	This implementation of the \textbf{CIR} algorithm leads to a worst-case complexity of $ \mathbf{O}(\vert\mathcal{C}\vert^3) ) $.
	\hfill $ \square $
\end{lemma}
\begin{proof}
	In each iteration we are summing ($\| $) a total of $ \vert\mathcal{E}\vert $ entries of $ M $. The lookup and insertion in an hash table are fast operations with complexity $ \mathbf{O}(1) $ which are each executed $ \vert\mathcal{C}\vert $ times. The $ \vert\mathcal{C}\vert $ strings that are used as key in the hash table have size proportional to $ rank(\mathcal{A}_i) $.\\ The complexity of the $ i^{th} $ iteration is then \mbox{$ \mathbf{O}(\vert\mathcal{E}\vert + \vert\mathcal{C}\vert + \vert\mathcal{C}\vert rank(\mathcal{A}_i)) $}. In the worst-case scenario the rank increases by one and the number of iterations is $ \mathbf{O}(\vert\mathcal{C}\vert) $. This implies total worst-case complexity of $ \mathbf{O}(\vert\mathcal{C}\vert^3) $.
\end{proof}
\begin{remark}
	In practice, the number of iterations seems to be much lower than $ \vert \mathcal{C}\vert $ which means that this is a very pessimistic upper bound for the complexity.
	\hfill $ \square $
\end{remark}
We illustrate this algorithm with the following example.
\begin{exmp}\label{exmp:cir_example}
	Consider the network illustrated in \cref{fig:cir_exmp_net} with cell type partition $ \type{} = \{\{1,2,5,6\},\{3,4\}\} $. The edge weight monoid is the same as in the parallel of resistors (\cref{exmp:resistor}). We assume that the arrows all represent values of $ 30 $. Note that the zero of the monoid is $ 0_{\mathcal{M}} = \infty $.
	\begin{figure}[h]
		\centering
		\input{cir_exmp_net.tikz}
		\caption{Network of \cref{exmp:cir_example} illustrating the CIR algorithm}
		\label{fig:cir_exmp_net}
	\end{figure}
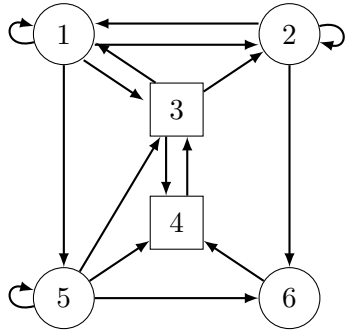
	This is represented by the matrix in \cref{eq:cir_exmp_net}.
	\begin{equation} M = 
	\begin{bmatrix}
	30     & 30     & 30     & \infty & \infty & \infty\\
	30     & 30     & 30     & \infty & \infty & \infty\\
	30     & \infty & \infty & 30     & 30     & \infty\\
	\infty & \infty & 30     & \infty & 30     & 30    \\
	30     & \infty & \infty & \infty & 30     & \infty\\
	\infty & 30     & \infty & \infty & 30     & \infty    
	\end{bmatrix}
	\label{eq:cir_exmp_net}
	\end{equation}
	If we are interested in finding the top partition $ \top $, we initialize $ \mathcal{A}_0 = \type{} $. This partition can be represented by the matrix $ P_0 $
	\begin{equation*} \mathcal{A}_0 = 
	\begin{bmatrix}
	1\\
	1\\
	2\\
	2\\
	1\\
	1  
	\end{bmatrix}
	\quad
	P_0 = 
	\begin{bmatrix}
	1 & 0\\
	1 & 0\\
	0 & 1\\
	0 & 1\\
	1 & 0\\
	1 & 0 
	\end{bmatrix}
	\end{equation*}
	Applying the algorithm we get 
	\begin{equation*}
	\left[\begin{array}{c|c} 
	\mathcal{A}_0 & M P_0
	\end{array}\right]
	=
	\left[\begin{array}{c|cc} 
	1 & 15 & 30\\
	1 & 15 & 30\\
	2 & 15 & 30\\
	2 & 15 & 30\\
	1 & 15 & \infty\\
	1 & 15 & \infty  
	\end{array}\right] 
	\end{equation*}
	whose row comparison determines the next iteration $ \mathcal{A}_1 $ and $ P_1 $
	\begin{equation*} \mathcal{A}_1 = 
	\begin{bmatrix}
	1\\
	1\\
	2\\
	2\\
	3\\
	3  
	\end{bmatrix}
	\quad
	P_1 = 
	\begin{bmatrix}
	1 & 0 & 0\\
	1 & 0 & 0\\
	0 & 1 & 0\\
	0 & 1 & 0\\
	0 & 0 & 1\\
	0 & 0 & 1
	\end{bmatrix}
	\end{equation*}
	Applying the same procedure
	\begin{equation*} 
	\left[\begin{array}{c|c} 
	\mathcal{A}_1 & M P_1
	\end{array}\right]
	= 
	\left[\begin{array}{c|ccc} 
	1 & 15     & 30     & \infty\\
	1 & 15     & 30     & \infty\\
	2 & 30     & 30     & 30\\
	2 & \infty & 30     & 15\\
	3 & 30     & \infty & 30\\
	3 & 30     & \infty & 30 
	\end{array}\right] 
	\end{equation*}
	and we get the second iteration defined by
	\begin{equation*} \mathcal{A}_2 = 
	\begin{bmatrix}
	1\\
	1\\
	2\\
	3\\
	4\\
	4  
	\end{bmatrix}
	\quad
	P_2 = 
	\begin{bmatrix}
	1 & 0 & 0 & 0\\
	1 & 0 & 0 & 0\\
	0 & 1 & 0 & 0\\
	0 & 0 & 1 & 0\\
	0 & 0 & 0 & 1\\
	0 & 0 & 0 & 1
	\end{bmatrix}
	\end{equation*}
	\begin{equation*} 
	\left[\begin{array}{c|c} 
	\mathcal{A}_2 & M P_2
	\end{array}\right]
	= 
	\left[\begin{array}{c|cccc} 
	1 & 15     & 30     & \infty & \infty\\
	1 & 15     & 30     & \infty & \infty\\
	2 & 30     & \infty & 30     & 30\\
	3 & \infty & 30     & \infty & 15\\
	4 & 30     & \infty & \infty & 30\\
	4 & 30     & \infty & \infty & 30
	\end{array}\right] 
	\end{equation*}
	We can now see that $ \mathcal{A}_2 = \mathcal{A}_3 $. This means that we have converged and $ \mathcal{A}_2 = cir(\mathcal{A}_0) = cir(\type{}) = \top $.\\
	This is not the only non-trivial balanced partition on this network. For example, with an initial partition \mbox{$ \mathcal{B}_0 = \{ \{1,2,5\},\{3,4\},\{6\} \}$} we find the other balanced partition $ \mathcal{B}_1 = cir(\mathcal{B}_0) =\{ \{1,2\},\{3\},\{4\},\{5\},\{6\} \}$.\\
	Note that we already knew that any other balanced partitions would have to be finer that $ \top $. Therefore we could have instead just verified if any of the partitions $ \{ \{1\},\{2\},\{3\},\{4\},\{5,6\} \} $ or $ \{ \{1,2\},\{3\},\{4\},\{5\},\{6\} \} $ were balanced.
	\hfill $ \square $
\end{exmp}

\section{Conclusion}
This paper generalizes the theory of coupled cell networks to multi edge and multi edge-type networks with arbitrarily complex edge weights. The formalism that was introduced here is simpler than the usual one based on groupoids of bijections of input sets. Moreover, we do not require the networks to obey such an artificial condition such as the `consistency condition'. We extend previous results about balanced partitions and invariant synchrony patterns to this more general setting.
An implementation of the CIR algorithm is presented which has a worst-case time complexity of \mbox{$ \mathbf{O}(\vert\mathcal{C}\vert^3) $} in opposition to the previous $ \mathbf{O}(\left(\vert\mathcal{E}\vert + \vert\mathcal{C}\vert\right)^4) $ cost.




\bibliographystyle{siamplain}
\bibliography{references}

\end{document}

%% file: edge_merginga.tikz
\begin{tikzpicture}[
node1/.style = {circle,minimum size=23,draw},
node2/.style = {circle,minimum size=23,draw,fill=white!75!black},
node3/.style = {circle,minimum size=23,draw,fill=white!50!black},
noderect/.style = {rectangle,minimum size=20,draw},
edge1/.style = {>=latex,thick},
edgedash/.style = {>=latex,thick,dashed},
edge2/.style = {>=latex,thick,blue},
edge3/.style = {>=latex,thick,red}
]
\node[noderect] at (-0.75,2)(n1){$ x_1 $};
\node[noderect] at (+0.75,2)(n2){$ x_2 $};
\node[node1] at (0,0)(n3){$ x_3 $};

\draw [->,edge1](n1) -- (n3);
\draw [->,edge1](n2) -- (n3);

\node (w13) at ($(n1)!0.5!(n3) + (-0.4,-0.0)$) {$ w_1 $};
\node (w23) at ($(n2)!0.5!(n3) + (+0.4,-0.0)$) {$ w_2 $};


\end{tikzpicture} 

%% file: edge_mergingb.tikz
\begin{tikzpicture}[
node1/.style = {circle,minimum size=23,draw},
node2/.style = {circle,minimum size=23,draw,fill=white!75!black},
node3/.style = {circle,minimum size=23,draw,fill=white!50!black},
noderect/.style = {rectangle,minimum size=20,draw},
edge1/.style = {>=latex,thick},
edgedash/.style = {>=latex,thick,dashed},
edge2/.style = {>=latex,thick,blue},
edge3/.style = {>=latex,thick,red}
]

\node[noderect] at (0,2)(n3){$ x_1=x_2 $};
\node[node1] at (0,0)(n12){$ x_3 $};

\draw [->,edge1](n3) -- (n12);
\node (w13) at ($(n3)!0.5!(n12) + (+0.8,-0.0)$) {$ w_1 \| w_2 $};


\end{tikzpicture} 

%% file: input_equiv1.tikz
\begin{tikzpicture}[
node1/.style = {circle,minimum size=23,draw},
node2/.style = {circle,minimum size=23,draw,fill=white!75!black},
node3/.style = {circle,minimum size=23,draw,fill=white!50!black},
noderect/.style = {rectangle,minimum size=20,draw},
nodetri/.style = {regular polygon, regular polygon sides=3 ,draw},
edge1/.style = {>=latex,thick},
edgedash/.style = {>=latex,thick,dashed},
edge2/.style = {>=latex,thick,blue},
edge3/.style = {>=latex,thick,red}
]
\node[noderect] at (-1.0,2)(n1){$ x_1 $};
\node[noderect] at (+0.0,2)(n2){$ x_1 $};
\node[node1] at (+1.0,2)(n3){$ x_2 $};
\node[node1] at (0,0)(nc){$ x $};

\draw [->,edge1](n1) -- (nc);
\draw [->,edge1](n2) -- (nc);
\draw [->,edge1](n3) -- (nc);

\node (w1c) at ($(n1)!0.5!(nc) + (-0.4,-0.0)$) {$ 1 $};
\node (w2c) at ($(n2)!0.5!(nc) + (+0.2,-0.0)$) {$ 1 $};
\node (w3c) at ($(n3)!0.5!(nc) + (+0.4,-0.0)$) {$ 3 $};


\end{tikzpicture} 

%% file: input_equiv2.tikz
\begin{tikzpicture}[
node1/.style = {circle,minimum size=23,draw},
node2/.style = {circle,minimum size=23,draw,fill=white!75!black},
node3/.style = {circle,minimum size=23,draw,fill=white!50!black},
noderect/.style = {rectangle,minimum size=20,draw},
nodetri/.style = {regular polygon, regular polygon sides=3 ,draw},
edge1/.style = {>=latex,thick},
edgedash/.style = {>=latex,thick,dashed},
edge2/.style = {>=latex,thick,blue},
edge3/.style = {>=latex,thick,red}
]
\node[node1] at (-1.0,2)(n1){$ x_2 $};
\node[noderect] at (+0.0,2)(n2){$ x_1 $};
\node[node1] at (+1.0,2)(n3){$ x_2 $};
\node[node1] at (0,0)(nc){$ x $};

\draw [->,edge1](n1) -- (nc);
\draw [->,edge1](n2) -- (nc);
\draw [->,edge1](n3) -- (nc);

\node (w1c) at ($(n1)!0.5!(nc) + (-0.4,-0.0)$) {$ 1 $};
\node (w2c) at ($(n2)!0.5!(nc) + (+0.2,-0.0)$) {$ 2 $};
\node (w3c) at ($(n3)!0.5!(nc) + (+0.4,-0.0)$) {$ 2 $};


\end{tikzpicture} 

%% file: input_equiv_compressed.tikz
\begin{tikzpicture}[
node1/.style = {circle,minimum size=23,draw},
node2/.style = {circle,minimum size=23,draw,fill=white!75!black},
node3/.style = {circle,minimum size=23,draw,fill=white!50!black},
noderect/.style = {rectangle,minimum size=20,draw},
nodetri/.style = {regular polygon, regular polygon sides=3 ,draw},
edge1/.style = {>=latex,thick},
edgedash/.style = {>=latex,thick,dashed},
edge2/.style = {>=latex,thick,blue},
edge3/.style = {>=latex,thick,red}
]
\node[node1] at (-0.75,2)(n1){$ x_2 $};
\node[noderect] at (+0.75,2)(n3){$ x_1 $};
\node[node1] at (0,0)(nc){$ x $};

\draw [->,edge1](n1) -- (nc);
\draw [->,edge1](n3) -- (nc);

\node (w1c) at ($(n1)!0.5!(nc) + (-0.4,-0.0)$) {$ 3 $};
\node (w3c) at ($(n3)!0.5!(nc) + (+0.4,-0.0)$) {$ 2 $};


\end{tikzpicture} 

%% file: Gadmissibility_example.tikz
\begin{tikzpicture}[
node1/.style = {circle,minimum size=23,draw},
node2/.style = {circle,minimum size=23,draw,fill=white!75!black},
node3/.style = {circle,minimum size=23,draw,fill=white!50!black},
noderect/.style = {rectangle,minimum size=20,draw},
edge1/.style = {>=latex,thick},
edge2/.style = {>=latex,thick,blue},
edge3/.style = {>=latex,thick,red}
]
\node[node1] at (0,0)(n1){1};
\node[node1] at (2,0)(n2){2};
\node[node1] at (1,{sqrt(2)})(n3){3};

\DoubleLine{n1}{n3}{<-,edge1}{}{->,edge1}{}
\draw [->,edge1](n1) -- (n2);
\DoubleLine{n2}{n3}{<-,edge1}{}{->,edge1}{}

\draw [->,edge1] (n1) edge[loop left,looseness=5] (n1);
\draw [->,edge1] (n3) edge[loop above,looseness=5] (n3);



\end{tikzpicture} 

%% file: Gadmissibility_example_color_part.tikz
\begin{tikzpicture}[
node1/.style = {circle,minimum size=23,draw},
node2/.style = {circle,minimum size=23,draw,fill=white!75!black},
node3/.style = {circle,minimum size=23,draw,fill=white!50!black},
noderect/.style = {rectangle,minimum size=20,draw},
edge1/.style = {>=latex,thick},
edge2/.style = {>=latex,thick,blue},
edge3/.style = {>=latex,thick,red}
]
\node[node2] at (0,0)(n1){1};
\node[node2] at (2,0)(n2){2};
\node[node1] at (1,{sqrt(2)})(n3){3};

\DoubleLine{n1}{n3}{<-,edge1}{}{->,edge1}{}
\draw [->,edge1](n1) -- (n2);
\DoubleLine{n2}{n3}{<-,edge1}{}{->,edge1}{}

\draw [->,edge1] (n1) edge[loop left,looseness=5] (n1);
\draw [->,edge1] (n3) edge[loop left,looseness=5] (n3);



\end{tikzpicture} 

%% file: Gadmissibility_example_quotient.tikz
\begin{tikzpicture}[
node1/.style = {circle,minimum size=23,draw},
node2/.style = {circle,minimum size=23,draw,fill=white!75!black},
node3/.style = {circle,minimum size=23,draw,fill=white!50!black},
noderect/.style = {rectangle,minimum size=20,draw},
edge1/.style = {>=latex,thick},
edge2/.style = {>=latex,thick,blue},
edge3/.style = {>=latex,thick,red}
]
\node[node2] at (0,0)(n1){1,2};
\node[node1] at (0,{sqrt(2)})(n3){3};

\DoubleLine{n1}{n3}{->,edge1}{2}{<-,edge1}{}

\draw [->,edge1] (n1) edge[loop left,looseness=5] (n1);
\draw [->,edge1] (n3) edge[loop left,looseness=5] (n3);



\end{tikzpicture} 

%% file: CCN_chain.tikz
\begin{tikzpicture}[
node1/.style = {circle,minimum size=23,draw},
node2/.style = {circle,minimum size=23,draw,fill=white!75!black},
node3/.style = {circle,minimum size=23,draw,fill=white!50!black},
noderect/.style = {rectangle,minimum size=20,draw},
edge1/.style = {>=latex,thick},
edge2/.style = {>=latex,thick,blue},
edge3/.style = {>=latex,thick,red}
]
\node[node1] at (0,0)(n1){1};
\node[node1] at (1.5,0)(n2){2};
\node[node1] at (3,0)(n3){3};

\draw [->,edge1](n1) -- (n2);
\draw [->,edge1](n2) -- (n3);




\end{tikzpicture} 

%% file: CCN_chain_unbalanced.tikz
\begin{tikzpicture}[
node1/.style = {circle,minimum size=23,draw},
node2/.style = {circle,minimum size=23,draw,fill=white!75!black},
node3/.style = {circle,minimum size=23,draw,fill=white!50!black},
noderect/.style = {rectangle,minimum size=20,draw},
edge1/.style = {>=latex,thick},
edge2/.style = {>=latex,thick,blue},
edge3/.style = {>=latex,thick,red}
]
\node[node1] at (0,0)(n1){1};
\node[node2] at (1.5,0)(n2){2};
\node[node2] at (3,0)(n3){3};

\draw [->,edge1](n1) -- (n2);
\draw [->,edge1](n2) -- (n3);




\end{tikzpicture} 

%% file: cir_exmp_net.tikz
\begin{tikzpicture}[
node1/.style = {circle,minimum size=23,draw},
node2/.style = {circle,minimum size=23,draw,fill=white!75!black},
node3/.style = {circle,minimum size=23,draw,fill=white!50!black},
noderect/.style = {rectangle,minimum size=20,draw},
edge1/.style = {>=latex,thick},
edge2/.style = {>=latex,thick,blue},
edge3/.style = {>=latex,thick,red}
]
\node[node1] at (0,0)(n1){1};
\node[node1] at (3,0)(n2){2};

\node[noderect] at (1.5,{-1})(n3){3};
\node[noderect] at (1.5,{-1-1.5})(n4){4};

\node[node1] at (0,{-2-1.5})(n5){5};
\node[node1] at (3,{-2-1.5})(n6){6};

\DoubleLine{n1}{n2}{<-,edge1}{}{->,edge1}{}
\draw [->,edge1] (n1) edge[loop left,looseness=5] (n1);
\draw [->,edge1] (n2) edge[loop right,looseness=5] (n2);
\draw [->,edge1] (n3) -- (n2);
\DoubleLine{n1}{n3}{<-,edge1}{}{->,edge1}{}
\DoubleLine{n3}{n4}{<-,edge1}{}{->,edge1}{}
\draw [->,edge1] (n5) -- (n4);
\draw [->,edge1] (n6) -- (n4);
\draw [->,edge1] (n5) -- (n3);
\draw [->,edge1] (n5) -- (n6);
\draw [->,edge1] (n1) -- (n5);
\draw [->,edge1] (n2) -- (n6);
\draw [->,edge1] (n5) edge[loop left,looseness=5] (n5);




\end{tikzpicture} 